\documentclass[12pt]{article}

\usepackage{fullpage}
\usepackage{amsmath}
\usepackage{amsthm}
\usepackage{amsfonts}
\usepackage{hyperref}

\newtheorem{lemma}{Lemma}
\newtheorem{claim}{Claim}
\newtheorem{corollary}{Corollary}
\newtheorem{theorem}{Theorem}

\newcounter{example}
\newenvironment{example}{\refstepcounter{example}\par\bigskip
\noindent\textit{Example~\theexample.} \rmfamily}{\hfill$\dashv$\bigskip}

\newcommand{\supp}{{\mathrm{Supp}}}
\newcommand{\tuples}{{\mathrm{Tup}}}
\newcommand{\domain}{{\mathrm{Dom}}}

\newcommand{\ltgc}{local-to-global consistency property}

\newcommand{\gcpr}{global consistency problem for relations}
\newcommand{\gcpb}{global consistency problem for bags}

\newcommand{\glcpb}{{\sc GCPB}}

\newcommand{\intcone}{\mathrm{intcone}}
\newcommand{\norm}[1]{\Vert #1 \Vert}
\newcommand{\bnorm}[1]{\norm{#1}_{\mathrm{b}}}
\newcommand{\unorm}[1]{\norm{#1}_{\mathrm{u}}}
\newcommand{\suppnorm}[1]{\norm{#1}_{\mathrm{supp}}}
\newcommand{\munorm}[1]{\norm{#1}_{\mathrm{mu}}}
\newcommand{\mbnorm}[1]{\norm{#1}_{\mathrm{mb}}}

\newcommand{\commentout}[1]{}

\usepackage{thmtools,thm-restate}

\title{\bf Structure and Complexity  of Bag Consistency}

\author{Albert Atserias\thanks{Universitat Polit\`ecnica de Catalunya,
    Barcelona, Catalonia, Spain. Atserias' research partially
    supported by MICIN project PID2019-109137GB-C22 (PROOFS).}  \and
  Phokion G. Kolaitis\thanks{UC Santa Cruz and IBM Research, Santa
    Cruz, CA, USA. Kolaitis' research partially supported by NSF Grant
    IIS-1814152.}}

\begin{document}

\maketitle

\begin{abstract}
  Since the early days of relational databases, it was realized that
  acyclic hypergraphs give rise to database schemas with desirable
  structural and algorithmic properties. In a by-now classical paper,
  Beeri, Fagin, Maier, and Yannakakis established several different
  equivalent characterizations of acyclicity; in particular, they
  showed that the sets of attributes of a schema form an acyclic
  hypergraph if and only if the local-to-global consistency property
  for relations over that schema holds, which means that every
  collection of pairwise consistent relations over the schema is
  globally consistent. Even though real-life databases consist of bags
  (multisets), there has not been a study of the interplay between
  local consistency and global consistency for bags.  We embark on
  such a study here and we first show that the sets of attributes of a
  schema form an acyclic hypergraph if and only if the local-to-global
  consistency property for bags over that schema holds. After this, we
  explore algorithmic aspects of global consistency for bags by
  analyzing the computational complexity of the global consistency
  problem for bags: given a collection of bags, are these bags
  globally consistent? We show that this problem is in NP, even when
  the schema is part of the input. We then establish the following
  dichotomy theorem for fixed schemas: if the schema is acyclic, then
  the global consistency problem for bags is solvable in polynomial
  time, while if the schema is cyclic, then the global consistency
  problem for bags is NP-complete. The latter result contrasts sharply
  with the state of affairs for relations, where, for each fixed
  schema, the global consistency problem for relations is solvable in
  polynomial time.
\end{abstract}

\section{Introduction} \label{sec:intro} Early investigations in
database theory led to the discovery that many fundamental algorithmic
problems about relational databases are intractable. In particular,
the relational join evaluation problem is NP-complete: given relations
$R_1,\ldots,R_m$ and a tuple $t$, does $t$ belong to the join
$R_1 \Join \ldots \Join R_m$ of the given relations? This motivated
the pursuit of tractable cases of the relational join evaluation
problem. In an influential paper \cite{DBLP:conf/vldb/Yannakakis81},
Yannakakis showed that the relational join evaluation problem is
solvable in polynomial time if the schema of the given relations is
acyclic, i.e., if the sets of the attributes of the given relations
are the hyperedges of an acyclic hypergraph. The notion of hypergraph
acyclicity turned out to have several other desirable properties in
relational databases that were explored in depth by Beeri, Fagin,
Maier, and Yannakakis \cite{BeeriFaginMaierYannakakis1983}. Arguably
the most prominent such property has to do with the universal relation
problem, also known as the global consistency problem
\cite{DBLP:journals/tods/AhoBU79,DBLP:conf/pods/Ullman82}. This
problem asks: given relations $R_1,\ldots,R_m$, is there a relation
$R$ such that, for every $i\leq m$, the projection of $R$ on the
attributes of $R_i$ is equal to $R_i$?  If the answer is positive,
then the relations $R_1,\ldots,R_m$ are said to be globally consistent
relations and $R$ is said to be a universal relation for them.
Honeyman, Ladner, and Yannakakis \cite{DBLP:journals/ipl/HoneymanLY80}
showed that the universal relation problem is NP-complete, even when
all input relations are binary. It is easy to see that if the
relations $R_1,\ldots,R_m$ are globally consistent, then they are
pairwise consistent, i.e., every two of them are globally consistent;
the converse, however, does not hold, in general.  Beeri et al.\
\cite{BeeriFaginMaierYannakakis1983} showed that a schema is acyclic
if and only if the local-to-global consistency property for relations
over that schema holds, which means that every collection of pairwise
consistent relations over the schema is globally consistent. Thus, for
acyclic schemas, pairwise consistency is both a necessary and
sufficient condition for global consistency; therefore, the universal
relation problem is solvable in polynomial time.

In all aforementioned results, relations are assumed to be sets. In
1993, Chaudhuri and Vardi \cite{DBLP:conf/pods/ChaudhuriV93} pointed
out that there is a gap between database theory and database practice
because ``real" databases use bags (multisets). They then called for a
re-examination of the foundations of databases where the fundamental
concepts and algorithmic problems are investigated under bag
semantics, instead of set semantics. In particular, Chaudhuri and
Vardi \cite{DBLP:conf/pods/ChaudhuriV93} raised the question of the
decidability of the conjunctive query containment problem under bags
semantics (the same problem under set semantics is known to be
NP-complete \cite{DBLP:conf/stoc/ChandraM77}). Various efforts in the
past and some recent progress notwithstanding
\cite{DBLP:conf/pods/KonstantinidisM19,DBLP:conf/pods/KhamisK0S20},
this question remains unanswered at present.

It is perhaps surprising that a study of consistency notions under bag
semantics has not been carried out to date. Our main goal in this
paper is to embark on such a study and to explore both structural and
algorithmic aspects of pairwise consistency and of global consistency
under bag semantics. In this study, the notions of consistency of bags
are, of course, defined using bag semantics in the computation of
projections.

In general, properties of relations do not automatically carry over to
similar properties of bags. This phenomenon manifests itself in the
context of consistency properties. Indeed, it is well known that if a
collection of relations is globally consistent, then their relational
join is a witness to their global consistency (see, e.g.,
\cite{DBLP:journals/ipl/HoneymanLY80}); in other words, their
relational join is a universal relation for them and, in fact, it is
the biggest universal relation. In contrast, as we point out in
Section \ref{sec:two-bag-cons}, this property fails for bags, i.e.,
there is a collection of bags that is globally consistent but the
bag-join of the bags in the collection is not a witness to their
global consistency.  In fact, this holds even for two consistent bags
and, furthermore, there may be no biggest witness to the consistency
of these bags.  Our first result establishes that two bags are
consistent if and only if they have the same projection on their
common attributes. While the analogous fact for relations is rather
trivial, here we need to bring in tools from the theory of linear
programming and maximum flow problems. As a corollary, we obtain a
polynomial-time algorithm for checking whether two given bags are
consistent and returning a witness to their consistency, if they are
consistent. After this, we establish our main result concerning the
structure of bag consistency. Specifically, we show that the sets of
attributes of a schema form an acyclic hypergraph if and only if the
local-to-global consistency for bags over that schema holds. This
shows that the main finding by Beeri et
al.\cite{BeeriFaginMaierYannakakis1983} about acyclicity and
consistency extends to bags. As we explain in Section
\ref{sec:three-bag-cons}, however, the architecture of the proof is
different from that in \cite{BeeriFaginMaierYannakakis1983}. In
particular, if a schema is cyclic, we give an explicit construction of
a collection of bags that are pairwise consistent, but not globally
consistent; the inspiration for our construction comes from an earlier
construction of hard-to-prove tautologies in propositional logic by
Tseitin \cite{Tseitin1968}.

We then explore algorithmic aspects of global consistency for bags by
analyzing the computational complexity of the global consistency
problem for bags: given a collection of bags, are these bags globally
consistent? Using a sparse-model property of integer programming that
is reminiscent of Carath\'eodory's Theorem for conic hulls
\cite{EisenbrandShmonin2006}, we first show that this problem is in
NP, even when the schema is part of the input. After this, we
establish the following dichotomy theorem for fixed schemas: if the
schema is acyclic, then the global consistency problem for bags is
solvable in polynomial time, while if the schema is cyclic, then the
global consistency problem for bags is NP-complete. The latter result
contrasts sharply with the state of affairs for relations, where, for
each fixed schema, the global consistency problem for relations is
solvable in polynomial time. Our NP-hardness results build on an
earlier NP-hardness result about three-dimensional statistical data
tables by Irving and Jerrum \cite{DBLP:journals/siamcomp/IrvingJ94},
which was later on refined by De Loera and Onn
\cite{DBLP:journals/siamcomp/LoeraO04}. Translated into our context,
this result asserts the NP-hardness of the global consistency problem
for bags over the triangle hypergraph, i.e., the hypergraph with
hyperedges of the form~$\{A_1,A_2\}, \{A_2,A_3\},
\{A_3,A_1\}$. Finally, we give a polynomial-time algorithm for the
following problem: given an acyclic schema and a collection of
pairwise consistent bags over that schema, construct a (small) witness
to their global consistency.  For this, we use Carath\'eodory's
classical theorem for conic hulls and the existence of strongly
polynomial algorithms for maximum flow problems (for the latter, see,
e.g.,~\cite{DBLP:conf/stoc/Orlin13}).

\paragraph{Related Work}
The interplay between local consistency and global consistency arises
naturally in several different settings. Already in 1962, Vorob'ev
\cite{vorob1962consistent} studied this interplay in the setting of
probability distributions and characterized the \ltgc~for probability
distributions in terms of a structural property of hypergraphs that
turned out to be equivalent to hypergraph acyclicity. It appears that
Beeri et al.\cite{BeeriFaginMaierYannakakis1983} were unaware of
Vorob'ev work, but later on Vorob'ev's work was cited in a survey of
database theory by Yannakakis
\cite{DBLP:journals/sigact/Yannakakis96}. In recent years, the
interplay between local consistency and global consistency has been
explored at great depth in the setting of quantum mechanics by
Abramsky and his collaborators (see, e.g.,
\cite{DBLP:journals/corr/abs-1102-0264,DBLP:conf/csl/AbramskyBKLM15,DBLP:journals/corr/abs-1111-3620}).
In that setting, the interest is in contextuality phenomena, which are
situations where collections of measurements are locally consistent
but globally inconsistent - the celebrated Bell's Theorem
\cite{bell1964einstein} is an instance of this. The similarities
between these different settings (probability distributions,
relational databases, and quantum mechanics) were pointed out
explicitly by Abramsky
\cite{DBLP:conf/birthday/Abramsky13,DBLP:journals/eatcs/Abramsky14}. This
also raised the question of developing a unifying framework in which,
among other things, the results by Vorob'ev and the results by Beeri
et al.\ are special cases of a single result. Using a relaxed notion
of consistency, we recently established such a result
for~$K$-relations, where~$K$ is a positive semiring
\cite{DBLP:journals/corr/abs-2009-09488}\footnote{This paper will
  appear in a forthcoming volume in honor of Samson Abramsky's
  contributions to logic.}. By definition, a~$K$-relation is a
relation such that each of its tuples has an associated element from
the semiring~$K$ as value. In particular, if~$\mathbb{Z}^{\geq 0}$ is
the semiring of non-negative integers (also known as the bag
semiring), then the~$\mathbb{Z}^{\geq 0}$-relations are precisely the
bags. For~$\mathbb{Z}^{\geq 0}$-relations, however, the relaxed notion
of consistency that we studied in
\cite{DBLP:journals/corr/abs-2009-09488} is essentially equivalent to
the consistency of probability distributions with rational
values. This left open the question of exploring the interplay between
(the standard notions of) local consistency and global consistency for
bags, which is what we set to do in the present paper. Furthermore, as
described earlier, here we also explore algorithmic aspects of global
consistency, which were not addressed at all in
\cite{DBLP:journals/corr/abs-2009-09488}.

\section{Preliminaries} \label{sec:prelims}

An \emph{attribute}~$A$ is a symbol with an associated
set~$\domain(A)$ called its \emph{domain}. If~$X$ is a finite set of
attributes, then we write~$\tuples(X)$ for the set
of~\emph{$X$-tuples}; this means that~$\tuples(X)$ is the set of
functions that take each attribute~$A \in X$ to an element of its
domain~$\domain(A)$. Note that~$\tuples(\emptyset)$ is non-empty as it
contains the \emph{empty tuple}, i.e., the unique function with empty
domain. If~$Y \subseteq X$ is a subset of attributes and~$t$ is
an~$X$-tuple, then the \emph{projection of~$t$ on~$Y$}, denoted
by~$t[Y]$, is the unique~$Y$-tuple that agrees with~$t$ on~$Y$. In
particular,~$t[\emptyset]$ is the empty tuple.

Let~$X$ be a set of attributes. We will view relations and bags
over~$X$ as functions from the set~$\tuples(X)$ to, respectively, the
Boolean semiring and the semiring of non-negative integers. The
\emph{Boolean semiring}~$\mathbb{B} = (\{0,1\},\vee,\wedge,0,1)$ has
disjunction~$\vee$ and conjunction~$\wedge$ as operations, and~$0$
(false) and~$1$ (true) as the identity elements of~$\vee$
and~$\wedge$. The
semiring~$\mathbb{Z}^{\geq 0}= (\{0,1,2,\ldots\},+, \times, 0, 1)$ of
non-negative integers has the standard arithmetic operations of
addition~$+$ and multiplication~$\times$, and~$0$ and~$1$ as the
identity elements of~$+$ and~$\times$.

A \emph{relation} over~$X$ is a
function~$R:\tuples(X)\rightarrow \{0,1\}$, while a \emph{bag}
over~$X$ is a function~$R:\tuples(X)\rightarrow \{0,1,2,\ldots\}$. We
write~$R(X)$ to emphasize the fact that~$R$ is a relation or a bag
over \emph{schema}~$X$.  If~$R$ is a relation or a bag, then the
\emph{support} of~$R$, denoted by~$\supp(R)$, is the set
of~$X$-tuples~$t$ that are assigned non-zero value, i.e.,
\begin{equation}
\supp(R) := \{ t \in \tuples(X) : R(t) \not= 0 \}. \label{def:support}
\end{equation}
Whenever no confusion arises, we write~$R'$ to denote~$\supp(R)$.  We
say that~$R$ is \emph{finite} if its support~$R'$ is a finite set. In
what follows, we will make the blanket assumption that all relations
and bags considered are finite, so we will omit the term ``finite".
Every relation~$R$ can be identified with its support~$R'$, thus every
relation~$R$ can be viewed as a finite set of~$X$-tuples. If~$R$ is a
bag and~$t$ is an~$X$-tuple, then the non-negative integer~$R(t)$ is
called the \emph{multiplicity} of~$t$ in~$R$; we will often
write~$t:R(t)$ to denote that the multiplicity of~$t$ in~$R$ is equal
to~$R(t)$.  Therefore, relations are bags in which the multiplicity of
each tuple is~$0$ or~$1$.  Every bag~$R$ can be viewed as a finite set
of elements of the form~$t: R(t)$, where~$t\in R'$. Thus,
if~$X=\{A,B\}$,
then~$R(A,B)= \{(a_1,b_1):2, (a_2,b_2): 1, (a_3,b_3):5\}$ represents
the bag~$R$ over~$X$ such that~$R(a_1,b_1)=
2$,~$R(a_2,b_2)= 1$,~$R(a_3,b_3)= 5$, and~$R(a,b) = 0$, for all other
pairs~$(a,b)$.  This bag can also be represented in tabular form as
follows:
\begin{center}
  \begin{tabular}{lll}
  $A$ & $B$  & \#  \\
  $a_1$ & $b_1:$ & 2 \\
   $a_2$ & $b_2:$ & 1 \\
   $a_3$ &  $b_3:$ & 5
  \end{tabular}
\end{center}
If~$R$ and~$S$ are two bags over the schema~$X$, then~$R$ is
\emph{bag-contained} in~$S$, denoted by~$R\subseteq_b S$,
if~$R(t)\leq S(t)$ for every~$X$-tuple~$t$.

Let~$R$ be a relation over~$X$ and assume that~$Z\subseteq X$.  The
\emph{projection of~$R$ on~$Z$}, denoted by~$R[Z]$, is the relation
over~$Z$ consisting of all projections~$t[Y]$ as~$t$ ranges over~$R$.

Let~$R$ be a bag over~$X$ and assume that~$Z\subseteq X$.  If~$t$ is
a~$Z$-tuple, then the \emph{marginal of~$R$ over~$t$} is defined by
\begin{equation} \label{eqn:marginal}
R(t) := \sum_{r \in R': \atop r[Z] = t} R(r).
\end{equation}
Thus, every bag~$R$ over~$X$ induces a bag over~$Z$, which is called
the \emph{marginal of~$R$ on~$Z$} and is denoted by~$R[Z]$. Note that
the preceding equation defines also the projection of a relation,
provided the sum is interpreted as the disjunction~$\vee$ over the
Boolean semiring.  It is easy to verify that the following facts hold
for every bag~$R$ over~$X$.
\begin{itemize} \itemsep=0pt
\item For all~$Z \subseteq X$, we have~$R'[Z] = R[Z]'$.
\item For all $W \subseteq Z \subseteq X$, we have $R[Z][W] = R[W]$.
\end{itemize}

If~$X$ and~$Y$ are sets of attributes, then we write~$XY$ as
shorthand for the union~$X \cup Y$. Accordingly, if~$x$ is
an~$X$-tuple and~$y$ is a~$Y$-tuple with the property
that~$x[X \cap Y] = y[X \cap Y]$, then we write~$xy$ to denote
the~$XY$-tuple that agrees with~$x$ on~$X$ and on~$y$ on~$Y$.  We
say that~\emph{$x$ joins with~$y$}, and that~\emph{$y$ joins
  with~$x$}, to \emph{produce} the tuple~$xy$.

If~$R$ is a relation over~$X$ and~$S$ is a relation over~$Y$, then
their \emph{join}~$R \Join S$ is the relation over~$XY$ consisting of
all tuples~$XY$-tuples~$t$ such that~$t[X]$ is in~$R$ and~$t[Y]$ is
in~$S$, i.e., all tuples of the form~$xy$ such that~$x\in
R'$,~$y\in R'$, and~$x$ joins with~$y$.  If~$R$ is a bag over~$X$
and~$S$ is a bag over~$Y$, then their \emph{bag join}~$R \Join_b S$ is
the bag over~$XY$ with support~$R'\Join S'$ and such that
every~$XY$-tuple~$t \in R'\Join S'$ has
multiplicity~$ (R\Join_b S)(t)= R(t[X])\times S(t[Y])$.

\section{Consistency of Two Bags} \label{sec:two-bag-cons}

We say that two relations~$R(X)$ and~$S(Y)$ are \emph{consistent} if
there exists a bag~$T(XY)$ with~$T[X] = R$ and~$T[Y] = S$.  Similarly,
we say that two bags~$R(X)$ and~$S(Y)$ are \emph{consistent} if there
exists a bag~$T(XY)$ with~$T[X] = R$ and~$T[Y] = S$, where now the
projections are computed according to Equation
(\ref{eqn:marginal}). In such a case, we say that~$T$ \emph{witnesses}
the consistency of~$R$ and~$S$.  A simple calculation shows that
if~$R(X)$ and~$S(Y)$ are consistent bags and~$T$ is a bag that
witnesses their consistency, then the support~$T'$ of~$T$ is a subset
of the join~$R' \Join S'$ of the supports.

\begin{lemma} \label{lem:inclusion} If~$R(X)$ and~$S(Y)$ are
  consistent bags and~$T(XY)$ is a bag that witnesses their
  consistency, then~$T' \subseteq R' \Join S'$.
\end{lemma}

\begin{proof} If~$t \in T'$, then~$T(t) \geq 1$, so~$R(t[X]) \geq 1$
  by~$R = T[X]$, and~$S(t[Y]) \geq 1$ by~$S = T[Y]$.
  Hence~$t[X] \in R'$ and~$t[Y] \in S'$, so~$t \in R' \Join S'$.
\end{proof}

If two relations~$R(X)$ and~$S(Y)$ are consistent, then their
join~$R\Join S$ witnesses their consistency; in fact,~$R\Join S$ is
the largest relation that has this property. In contrast, there are
consistent bags~$R(X)$ and~$S(Y)$ such that the support~$T'$ of every
bag~$T$ witnessing their consistency is a proper subset
of~$R'\Join S'$. An example of this is provided by the
bags~$R_1(AB) = \{ (1,2):{1},\; (2,2):{1} \}$
and~$S_1(BC) = \{ (2,1):{1},\; (2,2):{1} \}$; their consistency (as
bags) is witnessed by the
bags~$T_1(ABC) = \{ (1,2,2):{1},\; (2,2,1):{1} \}$
and~$T_2(ABC) = \{ (1,2,1):{1},\; (2,2,2):{1} \}$, but, as one can
easily verify, no other bag. This example can be extended as
follows. For~$n \geq 2$, let~$R_{n-1}(A,B)$ and~$S_{n-1}(B,C)$ be the
bags
  \begin{align*}
  & \{(1,2):{1},\; (2,2):{1},\;(1,3):{1},\;(3,3):{1},\; \ldots, \; (1,n):{1},(n,n):{1}\} \\
  & \{(2,1):{1},\; (2,2):{1},\;(3,1):{1},\;(3,3):{1},\; \ldots, \; (n,1):{1},(n,n):{1}\},
  \end{align*}
respectively.
%
%
For every~$n \geq 2$, the bags~$R_{n-1}$ and~$S_{n-1}$ are consistent
and there are exactly~$2^{n-1}$ bags witnessing their
consistency. Furthermore, these witnesses are pairwise incomparable in
the bag-containment sense and their supports are properly contained in
the support~$(R_{n-1}\Join_b S_{n-1})'$ of the bag
join~$R_{n-1}\Join_b S_{n-1}$. Note that the bags~$R_{n-1}$
and~$S_{n-1}$ are actually relations and that their
join~$R_{n-1} \Join S_{n-1}$ witnesses their consistency as relations,
but not as bags (where Equation (\ref{eqn:marginal}) is used to
compute the marginals).

With each pair of bags~$R(X)$ and~$S(Y)$, we associate the following
linear program~$P(R,S)$.  Let~$J = R' \Join S'$ be the join of the
supports of~$R$ and~$S$. For each~$t \in J$, there is a
variable~$x_t$. For each~$t \in J$ and~$r \in R'$,
define~$a_{r,t} = 1$ if~$t[X] = r$ and~$a_{r,t} = 0$
if~$t[X] \not= r$. Similarly, for each~$t \in J$ and~$s \in S'$,
define~$a_{s,t} = 1$ if~$t[Y] = s$ and~$a_{s,t} = 0$
if~$t[Y] \not= s$. The constraints of~$P(R,S)$~are:
\begin{equation}
\begin{array}{lll}
\sum_{t \in J} a_{r,t} x_t = R(r) & & \text{ for $r \in R'$, }\\
\sum_{t \in J} a_{s,t} x_t = S(s) & & \text{ for $s \in S'$, } \\
x_t \geq 0 & & \text{ for $t \in J$.}
\end{array} \label{eqn:lpfortwo}
\end{equation}
If we write the equations of~$P(R,S)$ in matrix form as~$A x = b$,
then the matrix~$A$ has special structure: its set of rows is
partitioned into two sets in such a way that every column has at most
one~$1$ entry in each part, and the rest of entries of the column
are~$0$.  This means that~$A$ is the vertex-edge incidence matrix of a
bipartite graph, so by Example~1 in Section~19.3 of Schrijver's book
\cite{Schrijver-book}, the matrix~$A$ is totally unimodular. By the
Hoffman-Kruskal Theorem (Corollary 19.2a in \cite{Schrijver-book}),
the polytope defined by~$P(R,S)$ is either empty or has integral
vertices.  Consequently,~$P(R,S)$ is feasible over the rationals if
and only if~$P(R,S)$ is feasible over the integers.  As we will soon
see, a different proof of this fact can be obtained using the
integrality theorem for max flow; for this, we will view~$P(R,S)$ as
the set of \emph{flow constraints} of an instance of the max-flow
problem, as we discuss next.

A network~$N = (V,E,c,s,t)$ is a directed graph~$G = (V,E)$ with a
non-negative weight~$c(u,v)$, called the \emph{capacity}, assigned to
each edge~$(u,v) \in E$, and two distinguished vertices~$s,t \in V$,
called the \emph{source} and the \emph{sink}.  A \emph{flow} for the
network is an assignment of non-negative weights~$f(u,v)$ on the
edges~$(u,v) \in E$ in such a way that both the capacity constraints
and the flow constraints are respected, i.e.,~$f(u,v) \leq c(u,v)$ for
each~$(u,v) \in E$,
and~$\sum_{v \in N^{-}(u)} f(v,u) = \sum_{w \in N^{+}(u)} f(u,w)$ for
each~$u \in V \setminus \{s,t\}$, where~$N^-(u)$ and~$N^+(u)$ denote
the sets of in-neighbors and out-neighbors of~$u$ in~$G$.  The
\emph{value} of such a flow is the
quantity~$\sum_{w \in N^+(s)}f(s,w) = \sum_{v \in N^-(t)} f(v,t)$,
where the equality follows from the flow constraints.  In the
\emph{max-flow problem}, the goal is to find a max flow, that is, a
flow of maximum value.  We say that a flow is \emph{saturated}
if~$f(s,w) = c(s,w)$ for every~$w \in N^+(s)$ and~$f(v,t) = c(v,t)$
for every~$v \in N^-(t)$.  It is obvious that if a saturated flow
exists, then every max flow is saturated.  The converse need not be
true,
unless~$\sum_{w \in N^+(s)} c(s,w) = \sum_{v \in N^-(t)} c(v,t)$.

With each pair~$R(X)$ and~$S(Y)$ of bags, we associate the following
network~$N(R,S)$.  The network has~$1+|R'|+|S'|+1$ vertices: one
source vertex~$s^*$, one vertex for each tuple~$r$ in the support~$R'$
of~$R$, one vertex for each tuple~$s$ in the support~$S'$ of~$S$, and
one target vertex~$t^*$. There is an arc of capacity~$R(r)$ from~$s^*$
to~$r$ for each~$r \in R'$, an arc of capacity~$S(s)$ from~$s$
to~$t^*$ for each~$s \in S'$, and an arc of
unbounded (i.e., very large) capacity from~$t[X]$ to~$t[Y]$ for
each~$t \in R' \Join S'$.

The next result yields several different characterization of the
consistency of two bags.

\begin{lemma} \label{lem:two-cons}
  Let $R(X)$ and $S(Y)$ be two bags. The following statements are
  equivalent:
  \begin{enumerate} \itemsep=0pt
  \item $R(X)$ and $S(Y)$ are consistent.
  \item $R[X \cap Y] = S[X \cap Y]$.
  \item $P(R,S)$ is feasible over the rationals.
  \item $P(R,S)$ is feasible over the integers.
  \item $N(R,S)$ admits a saturated flow.
  \end{enumerate}
\end{lemma}

\begin{proof}
  Let~$Z = X \cap Y$.  For (1) implies (2), assume that~$T$ witnesses
  the consistency of~$R$ and~$S$. Then~$T[X] = R$ and~$T[Y] = S$ and
  hence~$R[Z] = T[X][Z] = T[Z] = T[Y][Z] = S[Z]$. For (2) implies (3),
  assume that~$R[Z] = S[Z]$.  We show that~$P(R,S)$ is feasible over
  the rationals.  Let~$J = R' \Join S'$ and for each~$t \in J$
  set~$x_t := R(t[X])S(t[Y])/R(t[Z]) = R(t[X])S(t[Y])/S(t[Z])$, where
  the equality follows from the assumption that~$R[Z] = S[Z]$. For
  each fixed~$r \in R'$, let~$u = r[Z]$ and note
  that
  \begin{equation*}
  \sum_{t \in J} a_{r,t} x_t = (R(r)/S(u)) \sum_{t \in J:\atop
    t[X]=r} S(t[Y]) = (R(r)/S(u)) \sum_{s \in S':\atop s[Z]=u} S(s) = R(r).
  \end{equation*}
  For each fixed~$s \in R'$, let~$u = s[Z]$ and note
  that
  \begin{equation*}
  \sum_{t \in J} a_{s,t} x_t = (S(s)/R(u)) \sum_{t \in J:\atop
    t[Y]=s} R(t[X]) = (S(s)/R(u)) \sum_{r \in R':\atop r[Z]=u} R(s) = S(s).
  \end{equation*}
  Therefore, since~$x_t \geq 0$, we have shown that~$P(R,S)$ is
  feasible over the rationals.
  For (3) implies (5),
  let~$x^* = (x^*_t)_{t \in J}$ be a rational  solution
  for~$P(R,S)$ and let~$f$ be the following assignment for~$N(R,S)$:
  \begin{equation*}
  \begin{array}{lll}
  f(s^*,r) := c(s^*,r) = R(r) & & \text{ for each $r \in R'$;} \\
  f(t[X],t[Y]) := x^*_t & & \text{ for each $t\in J$;}\\
  f(s,t^*) := c(s,t^*) = S(s) & & \text{ for each $s \in S'$.}
  \end{array}
\end{equation*}
This assignment is a flow since the equations of~$P(R,S)$ say that the
flow-constraints are satisfied; furthermore, it is a saturated flow by
construction. For (5) implies (1), let~$g$ be a saturated flow
for~$N(R,S)$; in particular, this is a max flow for~$N(R,S)$. Since
all capacities in~$N(R,S)$ are integers, the integrality theorem for
the max-flow problem asserts that there is a max flow~$f$ consisting
of integers (see, e.g., \cite{DBLP:books/daglib/0007023}), which, of
course, is also a saturated flow.  Let~$T(XY)$ be the bag defined by
setting~$T(t) := f(t[X],t[Y])$ for each~$t \in R' \Join S'$. Since~$f$
is saturated, we have that~$f(s^*,r) = c(s^*,r) = R(r)$ for
each~$r \in R'$ and~$f(s,t^*) = c(s,t^*) = S(s)$ for each~$s \in
S'$. This means that the flow-constraints imply that~$T$ witnesses the
consistency of~$R$ and~$S$. Thus, we have established that statements
(1), (2), (3), and (5) are equivalent. The equivalence of statements
(1) and (4) is immediate from the definitions.
\end{proof}

The equivalence of statements (1) and (2) in Lemma \ref{lem:two-cons}
yields a simple polynomial-time test to determine the consistency of
two bags, namely, given two bags~$R(X)$ and~$S(y)$, check whether or
not~$R[X\cap Y] = S[X\cap Y]$. Furthermore, the equivalence of
statements (1) and (5) implies that there is a polynomial-time
algorithm for constructing a witness to the consistency of two
consistent bags. This is so, because it is well known that there are
polynomial-time algorithms for the max-flow problem. As a matter of
fact, there are strongly polynomial algorithms for this problem, such
as Orlin's algorithm \cite{DBLP:conf/stoc/Orlin13}, which finds a
maximum flow in time~$O(|V||E|)$. Thus, we have the following result.

\begin{corollary} \label{cor:two-cons-poly} There is a strongly
  polynomial-time algorithm that, given two bags, determines whether
  the bags are consistent and, if they are, constructs a bag
  witnessing their consistency.
\end{corollary}

We note that it is not known whether a strongly polynomial algorithm
for linear programming exists.
However, any algorithm for solving linear programming in time
polynomial in the bit-complexity of its data could be used to find a
witness to the consistency of two consistent bags. Simultaneously, the
algorithm could be asked to minimize any given linear function of the
multiplicities of the witnessing bag. Furthermore, it would accomplish
these tasks in time polynomial in the bit-complexity representation of
the input bags and the objective function. This follows from
Lemma~\ref{lem:two-cons} combined with the fact that, by the
Hoffman-Kruskal Theorem, all vertices of the polytope defined
by~$P(R,S)$ are integral.

\section{Consistency of Three or More Bags} \label{sec:three-bag-cons}

Let~$R_1(X_1),\ldots,R_m(X_m)$ be bags over the
schemas~$X_1,\ldots,X_m$. We say that the collection~$R_1,\ldots,R_m$
is \emph{globally consistent} if there is a bag~$T$
over~$X_1 \cup \cdots \cup X_m$ such that~$R_i = T[X_i]$ for
all~$i \in [m]$.  We say that such a bag \emph{witnesses} the global
consistency of~$R_1,\ldots,R_m$.  We also say that the
bags~$R_1,\ldots,R_m$ are \emph{pairwise consistent} if for
every~$i,j \in [m]$ we have that~$R_i[X_i]$ and~$R_j[X_j]$ are
consistent.

Corresponding notions of global consistency and pairwise consistency
can be defined for relations, the only difference being that
the~$T[X_i]$'s and the~$R_i[X_i]$'s are projections of relations,
instead of marginals of bags. The following facts are well known (see,
e.g., \cite{DBLP:journals/ipl/HoneymanLY80}):
\begin{itemize} \itemsep=0pt
\item If~$R_1,\ldots,R_m$ are relations and~$T$ is a relation
  witnessing the global consistency of the collection~$R_1,\ldots,R_m$,
  then~$T \subseteq R_1\Join \cdots \Join R_m$.
\item If~$R_1,\ldots,R_m$ are relations, then the
  collection~$R_1,\ldots,R_m$ is globally consistent if and only
  if~$(R_1\Join \cdots \Join R_m)[X_i]= R_i$ for all~$i = 1,\ldots,m$.
\end{itemize}
Consequently, if the collection $R_1,\ldots,R_m$ is globally
consistent, then the join $R_1\Join \cdots \Join R_m$ is the \emph{largest}
relation witnessing their consistency. As seen in Section \ref{sec:prelims}, there are bags that are consistent, but their consistency is not witnessed by their bag-join.

From the definitions, it follows that
if~$R_1,\ldots,R_m$ are globally consistent bags, then they are also
pairwise consistent. The converse, however, need not be true, in
general. In fact, the converse fails even for relations.
For example, the relations~$R(AB) =
\{00, 11\}$,~$S(BC) = \{01, 10\}$,~$T(AC)=
\{00,11\}$ are pairwise consistent but not globally consistent.
The interplay between pairwise consistency and global consistency of relations has been extensively studied in database theory. We summarize some of the main findings next.

Pairwise consistency is a necessary, but not
sufficient, condition for global consistency of relations.
Beeri, Fagin, Maier, and Yannakakis
\cite{BeeriFaginMaierYannakakis1983}
 characterized the set of schemas
for which pairwise consistency is a necessary and sufficient condition
for global consistency of relations. Their characterization involves notions from hypergraph theory that we now review.

\paragraph{Acyclic Hypergraphs} A \emph{hypergraph} is a
pair~$H = (V,E)$, where~$V$ is a set of \emph{vertices} and~$E$ is a
set of \emph{hyperedges}, each of which is a non-empty subset of~$V$.
Every collection~$X_1,\ldots,X_m$ of sets of attributes can be
identified with a hypergraph~$H=(V,E)$,
where~$V = X_1\cup \cdots \cup X_m$ and~$E
=\{X_1,\ldots,X_m\}$. Conversely, every hypergraph~$H = (V,E)$ gives
rise to a collection~$X_1,\ldots,X_m$ of sets of attributes,
where~$X_1,\dots,X_m$ are the hyperedges of~$H$. Thus, we can move
seamlessly from collections of sets of attributes to hypergraphs, and
vice versa.  The notion of an \emph{acyclic} hypergraph generalizes
the notion of an acyclic graph. Since we will not work directly with
the definition of an acyclic hypergraph, we refer the reader to
\cite{BeeriFaginMaierYannakakis1983} for the precise
definition. Instead, we focus on other notions that are equivalent to
hypergraph acyclicity and will be of interest to us in the sequel.

\paragraph{Conformal and Chordal Hypergraphs} The \emph{primal} graph
of a hypergraph~$H = (V,E)$ is the undirected graph that has~$V$ as
its set of vertices and has an edge between any two distinct vertices
that appear together in at least one hyperedge of~$H$. A
hypergraph~$H$ is \emph{conformal} if the set of vertices of every
clique (i.e., complete subgraph) of the primal graph of~$H$ is
contained in some hyperedge of~$H$.  A hypergraph~$H$ is
\emph{chordal} if its primal graph is chordal, that is, if every cycle
of length at least four of the primal graph of~$H$ has a chord.  To
illustrate these concepts, let~$V_n=\{A_1,\ldots,A_n\}$ be a set
of~$n$ vertices and consider the hypergraphs
\begin{eqnarray}
P_n &  = &  (V_n, \{ {A_1,A_2}\}, \ldots, \{A_{n-1},A_n\})  \label{path-hyper}\\
C_n & = & (V_n, \{A_1,A_2\}, \ldots, \{A_{n-1},A_n\}, \{A_n,A_1\}) \label{cycle-hyper}\\
H_n & = & (V_n, \{ V_n\setminus \{A_i\}: 1\leq i\leq n \}) \label{clique-hyper}
\end{eqnarray}
If~$n\geq 2$, then the hypergraph~$P_n$ is both conformal and chordal.
The hypergraph~$C_3 = H_3$ is chordal, but not conformal. For
every~$n\geq 4$, the hypergraph~$C_n$ is conformal, but not chordal,
while the hypergraph~$H_n$ is chordal, but not conformal.

\paragraph{Running Intersection Property} We say that a hypergraph~$H$
has the \emph{running intersection property} if there is a
listing~$X_1,\ldots,X_m$ of all hyperedges of~$H$ such that for
every~$i \in [m]$ with~$i \geq 2$, there exists a~$j < i$ such
that~$X_i \cap (X_1 \cup \cdots \cup X_{i-1}) \subseteq X_j$.

\paragraph{Join Tree} A \emph{join tree} for a hypergraph~$H$ is an
undirected tree~$T$ with the set~$E$ of the hyperedges of~$H$ as its
vertices and such that for every vertex~$v$ of~$H$, the set of
vertices of~$T$ containing~$v$ forms a subtree of~$T$, i.e., if~$v$
belongs to two vertices~$X_i$ and~$X_j$ of~$T$, then~$v$ belongs to
every vertex of~$T$ in the unique path from~$X_i$ to~$X_j$ in~$T$.

\paragraph{Local-to-Global Consistency Property for Relations}
Let~$H$ be a hypergraph and let~$X_1,\dots,X_m$ be a listing of all
hyperedges of~$H$. We say that~$H$ has the \emph{\ltgc~for
  relations} if every pairwise consistent
collection~$R_1(X_1),\ldots,R_m(X_m)$ of relations of
schema~$X_1,\ldots,X_m$ is globally consistent.

We are now ready to state the main result in Beeri et al.\
\cite{BeeriFaginMaierYannakakis1983}.

\begin{theorem} [Theorem 3.4 in \cite{BeeriFaginMaierYannakakis1983}] \label{thm:BFMY}
Let~$H$ be a hypergraph. The following statements are equivalent:
\begin{enumerate} \itemsep=0pt
\item[(a)] $H$ is an acyclic hypergraph.
\item[(b)] $H$ is a conformal and chordal hypergraph.
\item[(c)] $H$ has the running intersection property.
\item[(d)] $H$ has a join tree.
\item[(e)] $H$ has the \ltgc~for relations.
\end{enumerate}
\end{theorem}
As an illustration, if~$n\geq 2$, the hypergraph~$P_n$ is acyclic,
hence it has the \ltgc~for relations. In contrast, if~$n\geq 3$, the
hypergraphs~$C_n$ and~$H_n$ are cyclic, hence they do not have the
\ltgc~for relations.

In what follows, we will show that the preceding Theorem \ref{thm:BFMY} also holds for bags.  We need the following definition.   Let~$H$ be a hypergraph and let~$X_1,\dots,X_m$ be a listing of all
hyperedges of~$H$. We say that~$H$ has the \emph{\ltgc~for
  bags} if every pairwise consistent
collection~$R_1(X_1),\ldots,R_m(X_m)$ of relations of
schema~$X_1,\ldots,X_m$ is globally consistent.

\begin{theorem}   \label{thm:BFMY-bags}
Let~$H$ be a hypergraph. The following statements are equivalent:
\begin{enumerate} \itemsep=0pt
\item[(a)] $H$ is an acyclic hypergraph.
\item[(b)] $H$ is a conformal and chordal hypergraph.
\item[(c)] $H$ has the running intersection property.
\item[(d)] $H$ has a join tree.
\item[(e)] $H$ has the \ltgc~for bags.
\end{enumerate}
\end{theorem}

Before embarking on the proof of Theorem \ref{thm:BFMY-bags}, we need
some additional notions about hypergraphs and two technical lemmas.
We begin with the definitions of the notions needed.

Let~$H = (V,E)$ be a hypergraph. The \emph{reduction} of~$H$ is the
hypergraph~$R(H)$~whose set of vertices is~$V$ and whose hyperedges
are those hyperedges~$X \in E$ that are not included in any other
hyperedge of~$H$. A hypergraph~$H$ is \emph{reduced} if~$H=R(H)$.
If~$W \subseteq V$, then the \emph{hypergraph induced by~$W$ on~$H$}
is the hypergraph~$H[W]$~whose set of vertices is~$W$ and whose
hyperedges are the non-empty subsets of the form~$X \cap W$,
where~$X \in E$ is a hyperedge of~$H$; in
symbols,~$H[W] = (W, \{X \cap W: X \in E\}\setminus \{\emptyset\})$.

Let~$H = (V,E)$ be a hypergraph. For a vertex~$u \in V$, we
write~$H\setminus u$ for the hypergraph induced by~$V\setminus\{u\}$
on~$H$. For an edge~$e \in E$, we write~$H\setminus e$ for the
hypergraph with~$V$ as the set of its vertices and
with~$E \setminus \{e\}$ as the set of its edges. Let~$H' = (V',E')$
be another hypergraph.  We say that~$H'$ is obtained from~$H$ by a
\emph{vertex-deletion} if~$H' = H\setminus u$ for some~$u \in V$. We
say that~$H'$ is obtained from~$H$ by a \emph{covered-edge-deletion}
if~~$H' = H\setminus e$ for some~$e \in E$ such that~$e \subseteq f$
for some~$f \in E\setminus \{e\}$. In either case, we say that~$H'$ is
obtained from~$H$ by a \emph{safe-deletion operation}. We say that a
sequence of safe-deletion operations \emph{transforms~$H$ to~$H'$}
if~$H'$ can be obtained from~$H$ by starting with~$H$ and applying the
operations in order.

\begin{lemma} \label{lem:characconf} For every hypergraph~$H=(V,E)$
  the following statements hold:
\begin{enumerate}  \itemsep=0pt
\item $H$ is not chordal if and only if there
  exists~$W \subseteq V$~with $|W|\geq 4$ and such
  that~$R(H[W]) \cong C_n$, where~$n = |W|$.
\item $H$ is not conformal if and only if there
  exists~$W \subseteq V$~with $|W|\geq 3$ and such
  that~$R(H[W]) \cong H_n$, where~$n = |W|$.
\end{enumerate}
Moreover, there exist a polynomial-time algorithm that, given a
hypergraph~$H$ that is not chordal or not conformal, finds both a
set~$W$ as stated in (1) or (2) and a sequence of safe-deletion
operations that transforms~$H$ to~$R(H[W])$.
\end{lemma}

\begin{proof}
  The proof of (1) is straightforward. For the proof of (2)
  see~\cite{DBLP:journals/csur/Brault-Baron16}. Since there exist
  polynomial-time algorithms that test whether a graph is chordal
  (see, e.g.,~\cite{RoseTarjanLueker1976}), an algorithm to find a~$W$
  as stated in (1), when~$H$ is not chordal, is to iteratively delete
  vertices whose removal leaves a hypergraph with a non-chordal primal
  graph until no more vertices can be removed. Also, since there exist
  polynomial-time algorithms that test whether a hypergraph is
  conformal (see, e.g., Gilmore's Theorem in page~31
  of~\cite{Berge1989book}), an algorithm to find a~$W$ stated in (2),
  when~$H$ is not conformal, is to iteratively delete vertices whose
  removal leaves a non-conformal hypergraph until no more vertices can
  be removed. In both cases, once the set~$W$ is found, a sequence of
  safe-deletion operations that transforms~$H$~to~$R(H[W])$ if
  obtained by first deleting all vertices in~$V\setminus W$, and then
  deleting all covered edges.
\end{proof}

Let~$A_1,\ldots,A_n$ be attributes and let~$H = (V,E)$ be a hypergraph
with vertices~$V = \{A_1,\ldots,A_n\}$ and
edges~$E = \{X_1,\ldots,X_m\}$.  A \emph{collection of bags over~$H$}
is a collection~$D$ of bags~$R_1(X_1),\ldots,R_m(X_m)$ for
some~$m \geq 1$, i.e., each~$R_i$ is a bag over the schema~$X_i$.  For
an integer~$k \in [m]$, we say that a collection~$D$ of bags over~$H$
is~\emph{$k$-wise consistent} if for every~$I \subseteq [m]$
with~$|I|\leq k$, the collection~$\{R_i : i \in I\}$ is globally
consistent.  Observe that~$D$ is pairwise consistent if and only if it
is~$2$-wise consistent; furthermore,~$D$ is globally consistent if and
only if it is~$m$-wise consistent.

\begin{lemma} \label{lem:cons-preserv}
  Let~$H_0$ and~$H_1$ be hypergraphs such that~$H_0$ is obtained
  from~$H_1$ by a sequence of safe-deletion operations. For every
  collection~$D_0$ of bags over~$H_0$, there exists a collection~$D_1$ of bags
  over~$H_1$ such that, for every integer~$k \in [m]$, it holds
  that~$D_0$ is~$k$-wise consistent if and only if~$D_1$ is~$k$-wise
  consistent.  Moreover, there is a polynomial-time algorithm that,
  given~$D_0$ and a sequence of safe-deletion operations that
  transforms~$H_1$ to~$H_0$, computes~$D_1$.
\end{lemma}

\begin{proof}
  We first define the collection~$D_1$ in the case in which~$H_0$ is obtained from~$H_1$ by a
  single safe-deletion operation. In the case of a sequence of safe-deletion operations, the collection~$D_1$ in the statement of the lemma will be the
  result of iterating the construction in the first case~$t$ many times, where~$t$ is
  the number of operations that transforms~$H_1$ to~$H_0$. After the
  construction is spelled out, we analyse the run-time of the underlying
  algorithm and then prove its main property. In what follows,
  suppose that~$H_1 = (V_1,E_1)$, where~$V_1 = \{A_1,\ldots,A_n\}$
  and~$E_1 = \{X_1,\ldots,X_n\}$.

  Assume first that~$H_0 = H_1 \setminus X$ where~$X \in E_1$ is such
  that~$X \subseteq X_j$ for some~$j \in [m]$ with~$X \not= X_j$;
  i.e.,~$H_0$ is obtained from~$H_1$ by deleting a covered edge. In
  particular,~$V_0 = V_1$ and~$E_0 = E_1 \setminus \{X\}$. If the bags
  of~$D_0$ are~$S_i(X_i)$ for~$i \in [m]$ with~$X_i \not= X$,
  then~$D_1$ is defined as the collection with bags~$R_i(X_i)$
  for~$i \in [m]$ defined as follows: For each~$i \in [m]$,
  if~$X_i \not= X$, then~$R_i := S_i$; else let~$R_i := S_j[X]$.

  Assume next that~$H_0 = H_1 \setminus A$ where~$A \in V_1$;
  i.e.,~$H_0$ is obtained from~$H_1$ by deleting a vertex. In
  particular,~$V_0 = V_1\setminus\{A\}$ and~$E_0 = \{Y_1,\ldots,Y_m\}$
  where~$Y_i = X_i\setminus\{A\}$ for~$i = 1,\ldots,m$. Fix a default
  value~$u_0$ in the domain~$\mathrm{Dom}(A)$ of the attribute~$A$. If
  the bags of~$D_0$ are~$S_i(Y_i)$ for~$i \in [m]$, then~$D_1$ is
  defined as the collection with bags~$R_i(X_i)$ for~$i \in [m]$
  defined as follows: For each~$i \in [m]$, if~$A \not\in X_i$,
  let~$R_i := S_i$; else let~$R_i$ be the bag of
  schema~$X_i = Y_i \cup \{A\}$ defined for every~$X_i$-tuple~$t$
  by~$R_i(t) := 0$ if~$t(A) \not= u_0$ and~$R_i(t) := S_i(t[Y_i])$
  if~$t(A) = u_0$. We note that in case~$X_i = \{A\}$, the bag~$R_i$
  has empty schema~$Y_i = \emptyset$ and consists of the empty tuple
  with multiplicity~$S_i(u_0)$.

  It follows from the definitions that, in both cases, each bag~$R$
  of~$D_1$ has its  multiset cardinality bounded by~$S(\emptyset)$
  for some bag~$S$ of~$D_0$. In the case~$H_0 = H_1 \setminus X$, this
  follows from the fact that each bag of~$D_1$ is either a bag
  of~$D_0$ or the marginal of a bag of~$D_0$. In the
  case~$H_0 = H_1 \setminus A$, this follows from the fact that each
  bag of~$D_1$ is either a bag of~$D_0$ or a bag with the same
  multiset cardinality as a bag of~$D_0$. It follows by induction
  that if~$H_0$ is obtained from~$H_1$ by a sequence of~$t$ many
  safe-deletion operations, then the collection~$D_1$ of bags that results
  by  applying the construction~$t$ many times starting at~$D_0$ has
  each bag~$R$ of  multiset cardinality bounded by~$S(\emptyset)$
  for some bag~$S$ of~$D_0$. Thus,~$D_1$ has size at most~$t$ times
  the size of~$D_0$ and can be constructed in time polynomial in the
  size of~$D_0$ and the length~$t$ of the sequence.

  We prove the main property by cases. Fix an integer $k \geq 1$.

  \begin{claim}
    Assume~$H_0 = H_1 \setminus A$ for some vertex~$A \in V_1$. Then,
    the bags~$S_i(Y_i)$ of~$D_0$ are~$k$-wise consistent if and only
    if the bags~$R_i(X_i)$ of~$D_1$ are~$k$-wise consistent.
  \end{claim}

  \begin{proof}
    Fix~$I \subseteq [m]$ with~$|I| \leq k$,
    let~$X = \bigcup_{i \in I} X_i$ and~$Y = \bigcup_{i \in I}
    Y_i$. Observe that~$Y = X \setminus \{A\}$. In particular~$Y = X$
    if~$A$ is not in~$X$.

    (If): Let~$R$ be a bag over~$X$ that witnesses the consistency
    of~$\{ R_i : i \in I \}$, and let~$S := R[Y]$. We claim that~$S$
    witnesses the consistency of~$\{ S_i : i \in I \}$.  Indeed,
    $S[Y_i] = R[Y][Y_i] = R[Y_i] = R_i[Y_i] = S_i$, where the first
equality    follows from the choice of~$S$, the second equality follows
    from~$Y_i \subseteq Y$, the third equality follows from the facts
    that~$R[X_i] = R_i$ and~$Y_i \subseteq X_i$, and the fourth equality
    follows from the definition of~$R_i$.

    (Only if): Consider the two cases:~$A \not\in X$ or~$A \in X$.
    If~$A \not\in X$, then~$R_i = S_i$ for every~$i \in I$ and therefore
    the bags~$\{ R_i : i \in I\}$ are consistent because the
    bags~$\{ S_i : i \in I\}$ are consistent. If~$A \in X$, then
    let~$S$ be a bag over~$Y$ that witnesses the consistency of the
    bags~$\{ S_i : i \in I\}$, and let~$R$ be the bag over~$X$ defined
    for every~$X$-tuple~$t$ by~$R(t) := 0$ if~$t(A) \not= u_0$ and
    by~$R(t) := S(t[Y])$ if~$t(A) = u_0$. We claim that~$R$ witnesses
    the consistency of the bags~$R_i$ for~$i \in I$. We show
    that~$R_i = R[X_i]$ for~$i \in I$.  Towards this, first we argue
    that~$S[Y_i] = R[Y_i]$.  Indeed, for every~$Y_i$-tuple~$r$ we have
     \begin{align}
      S(r) = \sum_{s \in S' : \atop s[Y_i] = r} S(s) =
      \sum_{t \in \tuples(X) : \atop {t[Y_i] = r, \atop t(A) = u_0}} S(t[Y]) =
      \sum_{t \in S' : \atop t[Y_i] = r} R(t) =
      R(r), \label{eqn:lsllsggg}
    \end{align}
    where the first equality follows from~\eqref{eqn:marginal}, the
    second equality follows from the fact that the map~$t \mapsto t[Y]$ is a
    bijection between the set of~$X$-tuples~$t$ such that~$t[Y_i]=r$
    and~$t(A) = u_0$ and the set of~$Y$-tuples~$s$ such
    that~$s[Y_i]=r$, the third equality follows from the definition of~$R$, and
    the fourth equality follows from~\eqref{eqn:marginal}.

    In case~$A \not\in X_i$, we have that~$Y_i = X_i$, hence
    Equation~\eqref{eqn:lsllsggg} already shows
    that~$R_i = S_i = S[Y_i] = R[Y_i] = R[X_i]$. In case~$A \in X_i$,
    we use the fact that~$S_i = S[Y_i]$ to show that~$R_i =
    R[X_i]$. For every~$X_i$-tuple~$r$ with~$r(A) \not= u_0$, we
    have~$R_i(r) = 0$ and also~$R(r) = \sum_{t : t[X_i] = r} R(t) = 0$
    since~$t[X_i] = r$ and~$A \in X_i$
    implies~$t(A) = r(A) \not= u_0$. Thus, $R_i(r) = 0 = R(r)$ in this
    case. For every~$X_i$-tuple~$r$ with~$r(A) = u_0$, we have
    \begin{align}
      R_i(r) =  S_i(r[Y_i]) =  S(r[Y_i]) =  R(r[Y_i]), \label{lem:inter}
    \end{align}
    where the first equality follows from the definition of~$R_i$ and
    the assumption that~$r(A) = u_0$, the second equality follows from
    $S_i = S[Y_i]$, and the third equality follows from~\eqref{eqn:lsllsggg}.
    Continuing from the right-hand side of~\eqref{lem:inter}, we have
    \begin{align}
     R(r[Y_i]) =  \sum_{t \in R' : \atop t[Y_i] = r[Y_i]} R(t) =
     \sum_{t \in R' : \atop t[X_i] = r} R(t) = R(r),
    \label{lem:retni}
    \end{align}
    where the first equality follows from~\eqref{eqn:marginal}, the
    second equality follows from the assumption that~$A \in X_i$
    and~$r(A) = u_0$ together with~$R(t) = 0$ in
    case~$t(A) \not= u_0$, and the third equality follows
    from~\eqref{eqn:marginal}.  Combining~\eqref{lem:inter}
    with~\eqref{lem:retni}, we get~$R_i(r) = R(r)$ also in this
    case. This proves that~$R_i = R[X_i]$.
    \end{proof}

    \begin{claim}
      Assume~$H_0 = H_1 \setminus X$ for some edge~$X \in E_1$ that is
      covered in~$H_1$. Then, the bags~$S_i(X_i)$ of~$D_0$
      are~$k$-wise consistent if and only if the bags~$R_i(Y_i)$
      of~$D_1$ are~$k$-wise consistent.
    \end{claim}

     \begin{proof}
       Let~$l \in [m]$ be such that~$X = X_l \subseteq X_j$
       for some~$j \in [m]\setminus\{l\}$,
       so~$E_0 = \{ X_i : i \in [m]\setminus \{l\}\}$.

       (If): Fix~$I \subseteq [m]\setminus\{l\}$ with~$|I| \leq k$ and
       let~$X = \bigcup_{i \in I} X_i$. Let~$R$ be a bag over~$X$ that
       witnesses the consistency of~$\{ R_i : i \in I\}$ and
       let~$S = R$. Since~$S_i = R_i$ for
       every~$i \in [m]\setminus\{l\}$, it is obvious that~$S$
       witnesses the consistency of~$\{ S_i : i \in I \}$.

       (Only if): Fix~$I \subseteq [m]$ with~$|I| \leq k$ and
       let~$X = \bigcup_{i \in I} X_i$. Let~$S$ be a bag over~$X$ that
       witnesses the consistency
       of~$\{ S_i : i \in I\setminus\{l\} \}$ and let~$R = S$. We
       have~$R_l = S_j[X_l] = S[X_j][X_l] = R[X_j][X_l] = R[X_l]$
       where the first equality follows from the definition of~$R_l$, the
       second equality follows from the fact that~$S_j = S[X_j]$, the third
     equality  follows from the choice of~$R$, and the fourth equality follows
       from~$X_l \subseteq X_j$.
     \end{proof}

     The proof of  Lemma \ref{lem:cons-preserv} is now complete.
\end{proof}

Lemma \ref{lem:cons-preserv} implies  that the local-to-global
consistency property for bags is preserved under induced
hypergraphs and under reductions.

\begin{restatable}{corollary}{preservelemma} \label{lem:preserv1}
 If a hypergraph $H$ has the \ltgc~for bags, then for every subset~$W$ of the set of
  vertices of~$H$, the hypergraph~$R(H[W])$ also has the
  \ltgc~for~bags.
\end{restatable}

We are now ready to give the proof of Theorem \ref{thm:BFMY-bags}.

\begin{proof} [Proof of Theorem \ref{thm:BFMY-bags}]
  Let~$H$ be a hypergraph.  By Theorem \ref{thm:BFMY}, statements (a),
  (b), (c), and (d) are equivalent, because these statements express
  ``structural" properties of hypergraphs, i.e., their definitions
  involve only the vertices and the hyperedges of the hypergraph at
  hand.  So, we only have to show that statement (e), which involves
  ``semantics" notions about bags, is equivalent to (one of) the other
  three statements.  This will be achieved in two steps. First, we
  will show that if~$H$ has the running intersection property,
  then~$H$ has the \ltgc~for bags. Second, we will show that if~$H$ is
  not conformal or~$H$ is not chordal, then~$H$ does not have the
  \ltgc~for bags.

  \paragraph{Step 1.} Assume that the hypergraph~$H$ has the running
  intersection property. Hence, there is a listing~$X_1,\ldots,X_m$ of
  its hyperedges such that for every~$i \in [m]$ with~$i \geq 2$,
  there is a~$j \in [i-1]$ such
  that~$X_i \cap (X_1 \cup \cdots \cup X_{i-1}) \subseteq
  X_j$. Let~$R_1(X_1),\ldots,R_m(X_m)$ be a collection of pairwise
  consistent bags over the schemas~$X_1,\ldots,X_m$.  By induction
  on~$i = 1,\ldots,m$, we show that there is a bag~$T_i$
  over~$X_1 \cup \cdots \cup X_i$ that witnesses the global
  consistency of the bags~$R_1,\ldots,R_i$.  For~$i = 1$ the claim is
  obvious since~$T_1 = R_1$. Assume then that~$i \geq 2$ and that the
  claim is true for all smaller indices.
  Let~$X := X_1 \cup \cdots \cup X_{i-1}$ and, by the running
  intersection property, let~$j \in [i-1]$ be such
  that~$X_i \cap X \subseteq X_j$.  By induction hypothesis, there is
  a bag~$T_{i-1}$ over~$X$ that witnesses the global consistency
  of~$R_1,\ldots,R_{i-1}$.  First, we show that~$T_{i-1}$ and~$R_i$
  are consistent. By Lemma~\ref{lem:two-cons}, it suffices to show
  that~$T_{i-1}[X \cap X_i] = R_i[X \cap X_i]$. Let~$Z = X \cap X_i$,
  so~$Z \subseteq X_j$ by the choice of~$j$, and
  indeed~$Z = X_j \cap X_i$.  Since~$j \leq i-1$, we
  have~$R_j = T_{i-1}[X_j]$. Since~$Z\subseteq X_j$, we
  have~$R_j[Z] = T_{i-1}[X_j][Z] = T_{i-1}[Z]$. By assumption,
  also~$R_j$ and~$R_i$ are consistent, and~$Z = X_j \cap X_i$, which
  by Lemma~\ref{lem:two-cons} implies~$R_j[Z] = R_i[Z]$. By
  transitivity, we get~$T_{i-1}[Z] = R_i[Z]$, hence,
  by~Lemma~\ref{lem:two-cons}, the bags~$T_{i-1}$ and~$R_i$ are
  consistent.  Let~$T_i$ be a bag that witnesses the consistency of
  the bags~$T_{i-1}$ and~$R_i$.  We show that~$T_i$ witnesses the
  global consistency of~$R_1,\ldots,R_i$. Since~$T_{i-1}$ and~$R_i$
  are consistent, first note that~$T_{i-1} = T_i[X]$
  and~$R_i = T_i[X_i]$ by Lemma~\ref{lem:two-cons}. Now
  fix~$k \leq i-1$ and note that
  \begin{equation}
  R_k = T_{i-1}[X_k] = T_i[X][X_k] = T_i[X_k],
\end{equation}
where the first equality follows from the fact that~$T_{i-1}$ witnesses the
consistency of~$R_1,\ldots,R_{i-1}$ and~$k \leq i-1$, and  the other two equalities
follow from~$T_{i-1} = T_i[X]$ and the fact that~$X_k\subseteq X$.
Thus,~$T_i$ witnesses the consistency of~$R_1,\ldots,R_i$, which was
to be shown.

\paragraph{Step 2.}
Assume that the hypergraph~$H$ is not conformal or it is not chordal. By Lemma \ref{lem:characconf}, there is a subset~$W$ of~$V$ such that~$|W| \geq 3$
  and~$R(H[W])=(W,\{W\setminus\{A\} : A \in W\})$ or there is a subset~$W$ of~$V$ such that~$|W| \geq 4$ and~$R(H[W]) = (W,\{\{A_i,A_{i+1}\} : i \in [n]\})$, where~$A_1,\ldots,A_n$
  is an enumeration of~$W$ and~$A_{n+1} := A_1$.
  By Corollary \ref{lem:preserv1}, if~$H$ has the \ltgc~for bags, then
  for every subset~$W$ of~$V$, the hypergraph~$R(H[W])$ also has the
  \ltgc~for bags.  It follows that, to show that~$H$ does not have the
  \ltgc~for bags, it suffices to show that no hypergraph of the
  form~$(W,\{W\setminus\{A\} : A \in W\})$ with~$|W|\geq 3$ has the
  \ltgc~for bags, and no hypergraph of the
  form~$(W,\{\{A_i,A_{i+1}\} : i \in [n]\})$,
  where~$|W|\geq 4$,~$A_1,\ldots,A_n$ is an enumeration of~$W$,
  and~$A_{n+1} := A_1$ has the \ltgc~for bags.

  The preceding ``minimal" non-conformal and non-chordal hypergraphs
  share the following properties: 1)~all their hyperedges have the
  same number of vertices, and 2)~all their vertices appear in the
  same number of hyperedges. For hypergraphs~$H^*$ that have these
  properties, we construct a collection~$C(H^*)$ of~bags that are
  indexed by the hyperedges of~$H^*$

Let~$H^* = (V^*,E^*)$ be a hypergraph and let~$d$ and~$k$ be positive
integers. The hypergraph~$H^*$ is called~\emph{$k$-uniform} if every
hyperedge of~$H^*$ has exactly~$k$ vertices. It is called
\emph{$d$-regular} if any vertex of~$H^*$ appears in exactly~$d$
hyperedges of~$H$.
Thus, the ``minimal"   non-conformal hypergraph in
Lemma~\ref{lem:characconf}
is~$k$-uniform and~$d$-regular
for~$k := d := |W|-1$. Likewise, the ``minimal" non-chordal hypergraph in the same lemma
is~$k$-uniform and~$d$-regular for~$k := d := 2$.  For
each~$k$-uniform and~$d$-regular hypergraph~$H^*$ with~$d \geq 2$ and
with hyperedges~$E^* = \{X_1,\ldots,X_m\}$, we construct a
collection~$C(H^*) := \{ R_1(X_1),\ldots,R_m(X_m) \}$ of bags,
where~$R_i$ is a bag with~$X_i$ as its set of attributes. The
collection~$C(H^*)$ of these bags will turn out to be pairwise
consistent but not globally consistent.

For each~$i \in [m]$ with~$i \not= m$, let~$R_i$ be the unique bag
over~$X_i$ defined as follows: (a) the support~$R_i'$ of~$R_i$
consists of all tuples~$t : X_i \rightarrow \{0,\ldots,d-1\}$ whose
total sum~$\sum_{C \in X_i} t(C)$ is congruent to~$0$ mod~$d$;
(b)~$R_i(t) := 1$ for each such~$X_i$-tuple, and~$R_i(t) := 0$ for
every other~$X_i$-tuple.  For~$i = m$, let~$R_m$ be the unique bag
over~$X_m$ defined as follows: (a) the support~$R_m'$ of~$R_m$
consists of all tuples~$t : X_m \rightarrow \{0,\ldots,d-1\}$ whose
total sum~$\sum_{C \in X_m} t(C)$ is congruent to~$1$ mod~$d$;
(b)~$R_m(t) :=1$ for each such~$X_m$-tuple, and~$R_m(t) := 0$ for
every other~$X_m$-tuple.

By Lemma~\ref{lem:two-cons}, to show that the bags~$R_1,\ldots,R_m$
are pairwise consistent, it suffices to show that for every two
distinct~$i,j \in [m]$, we have~$R_i[Z] \equiv R_j[Z]$,
where~$Z := X_i \cap X_j$. In turn, this follows from the claim that
for every~$Z$-tuple~$t : Z \rightarrow \{0,\ldots,d-1\}$, we
have~$R_i(t) = R_j(t) = d^{k-|Z|-1}$. Indeed, since by~$k$-uniformity
every hyperedge of~$H$ has exactly~$k$ vertices, for
every~$u \in \{0,\ldots,d-1\}$, there are exactly~$d^{k-|Z|-1}$
many~$X_i$-tuples~$t_{i,u,1},\ldots,t_{i,u,d^{k-|Z|-1}}$ that
extend~$t$ and have total sum congruent to~$u$ mod~$d$. It follows
then that~$R_i[Z] = R_j[Z]$ regardless of whether~$n \in \{i,j\}$
or~$n \not\in \{i,j\}$, and hence any two~$R_i$ and~$R_j$ are
consistent by Lemma~\ref{lem:two-cons}.  To argue that the
relations~$R_1,\ldots,R_m$ are not globally consistent, we proceed by
contradiction. If~$R$ were a bag that witnesses their consistency,
then it would be non-empty and its support would contain a tuple~$t$
such that the projections~$t[X_i]$ belong to the supports~$R'_i$ of
the~$R_i$, for each~$i \in [m]$. In turn this means that
  \begin{align}
    & \textstyle{\sum_{C \in X_i} t(C) \;\equiv\; 0 \text{ mod } d}, \;\;\;\;\;\
 \text{ for $i \not = m$ } \label{eqn:restlu} \\
    & \textstyle{\sum_{C \in X_i} t(C) \;\equiv\; 1 \text{ mod } d}, \;\;\;\;\;\
 \text{ for $i = m$. } \label{eqn:a0lu}
  \end{align}
  Since by~$d$-regularity each~$C \in V$ belongs to exactly~$d$ many
  sets~$X_i$, adding up all the equations in~\eqref{eqn:restlu}
  and~\eqref{eqn:a0lu} gives
\begin{equation}
    \textstyle{\sum_{C \in V} dt(C) \;\equiv\; 1 \text{ mod } d},
  \end{equation}
  which is absurd since the left-hand side is congruent to~$0$
  mod~$d$, the right-hand side is congruent to~$1$ mod~$d$,
  and~$d \geq 2$ by assumption. This completes the proof of
  Theorem~\ref{thm:BFMY-bags}.
  \end{proof}

  Note that Beeri et al.\ \cite{BeeriFaginMaierYannakakis1983} showed
  that hypergraph acyclicity is equivalent to several other
  ``structural" properties of hypergraphs, such as
  Graham's algorithm succeeding on~$H$. We chose not to mention these other
  ``structural" properties here because we made no use of them in the
  proof of Theorem \ref{thm:BFMY-bags}; these properties, of course,
  can be added to the list of equivalent statements in Theorem
  \ref{thm:BFMY-bags}.  However, Beeri et al.\
  \cite{BeeriFaginMaierYannakakis1983} showed that hypergraph
  acyclicity is also equivalent to several ``semantic" properties of
  relations other than the \ltgc~for relations, including the
  existence of a full reducer for relations. As we shall discuss in
  Section \ref{sec:conclusions}, it remains an open problem to
  formulate a suitable concept of a full reducer for bags and show
  that the existence of such a full reducer for bags is equivalent to
  hypergraph acyclicity and, hence to the \ltgc~property for bags.
  The main technical obstacle is that the bag-join of a globally
  consistent collection of bags need not witness their global
  consistency.

  It should also be pointed out that the proof of
  Theorem~\ref{thm:BFMY} in \cite{BeeriFaginMaierYannakakis1983} has a
  different architecture than the proof of our
  Theorem~\ref{thm:BFMY-bags}. In particular, in proving the
  equivalence between the \ltgc~for relations and acyclicity, they
  make use of Graham's algorithm.

\section{Complexity of Bag Consistency} \label{sec:complexity}

In this section, we explore the algorithmic aspects of global
consistency.  We first discuss known results about global consistency
for relations.

\subsection{The Set Case}

The \emph{\gcpr}~asks: given a hypergraph~$H=(V,\{X_1,\ldots,X_m\})$
and relations~$R_1,\ldots,R_m$ over~$H$, are the
relations~$R_1,\ldots,R_m$ globally consistent?  This problem is also
known as the \emph{universal relation problem} since a relation~$W$
witnessing the global consistency of~$R_1,\ldots,R_m$ is called a
\emph{universal} relation for~$R_1,\ldots,R_m$. Honeyman, Ladner, and
Yannakakis \cite{DBLP:journals/ipl/HoneymanLY80} showed that the
\gcpr~is NP-complete. The proof of NP-hardness is a reduction from
{\sc 3-Colorability} in which each relation is binary and consists of
just six pairs. The proof of membership in NP uses the observation
that if a collection~$R_1,\ldots,R_m$ of relations is globally
consistent,
then a witness~$W$~of this fact can be obtained as follows: for
each~$i\leq m$ and each tuple~$t\in R_i$, pick a tuple in the
join~$R_1 \Join \cdots \Join R_m$ that extends~$t$ and insert it
in~$W$.  In particular, the cardinality~$|W|$ of~$W$ is bounded by the
sum~$\sum_{i=1}^m |R_i| \leq m \max\{|R_i|:i\in [m]\}$, and thus the
size of~$W$ is bounded by a polynomial in the size of the input
hypergraph~$H$ and the input relations~$R_1,\ldots,R_m$.

The main result in Beeri et al.\ \cite{BeeriFaginMaierYannakakis1983}
(stated here as Theorem \ref{thm:BFMY}) implies that the \gcpr~is
solvable in polynomial time when restricted to acyclic hypergraphs,
since, in this case, global consistency of relations is equivalent to
pairwise consistency of relations. Furthermore, for every fixed
hypergraph~$H=(V,\{X_1,\ldots,X_m\})$ (be it cyclic or acyclic), the
\gcpr~restricted to relations~$R_1,\ldots,R_m$ of
schemas~$X_1,\ldots,X_m$ is also solvable in polynomial time, since
one can first compute the join~$J = R_1\Join \cdots \Join R_m$ in
polynomial time and then check whether~$J[X_i]= R_i$ holds,
for~$i = 1,\ldots,m$. While the cardinality~$|J|$ of this witness~$J$
can only be bounded
by~$\prod_{i=1}^m |R_i| \leq \max\{|R_i|:i\in[m]\}^m$, this
cardinality is still polynomial in the size of the input because, in this
case, the exponent~$m$ is fixed and not part of the input.

\subsection{Decision Problem for Bags}

We now consider the \emph{\gcpb}, which asks: given a
hypergraph~$H=(V,\{X_1,\ldots,X_m\})$ and bags~$R_1,\ldots,R_m$
over~$H$, are the bags~$R_1,\ldots,R_m$ globally consistent?
We  also consider a family of decision problems arising from fixed hypergraphs.  Specifically, with every
fixed hypergraph~$H=(V,\{X_1,\ldots,X_m\})$, we associate  the decision
problem \glcpb$(H)$, which asks: given bags~$R_1,\ldots,R_m$ over~$H$,
are the bags~$R_1,\ldots,R_m$ globally consistent?

The first result
we obtain about the \gcpb\ is that it is in~NP, even if the
multiplicities of the tuples in the bags are represented in binary. To
prove this, we will show that if~$R_1,\ldots,R_m$ are globally
consistent bags, then there exists a bag~$W$ that witnesses their
global consistency and has size polynomial in the size
of~$R_1,\ldots,R_m$. More precisely, we will establish that the
support~$W'$ of the bag~$W$ has cardinality at
most~$\sum_{i=1}^m \sum_{r \in R'_i} \log(R_i(r)+1)$, and each
tuple~$t \in W'$ has multiplicity~$W(t)$ bounded
by~$\max\{ R_i(r) : i \in [m], r \in R'_i \}$. In order to establish
this, we need an integral version of Carath\'eodory's Theorem due to
Eisenbrand and Shmonin \cite{EisenbrandShmonin2006}. For a finite
set~$X \subseteq \mathbb{R}^d$ of real vectors,
let~$\mathrm{intcone}(X)$ denote the \emph{integer conic hull} of~$X$,
that is, the set of all vectors of the
form~$c_1 x_1 + \cdots + c_t x_t$, where~$c_1,\ldots,c_t$ are
non-negative integers and~$x_1,\ldots,x_t$ are vectors in~$X$.

\begin{lemma}[Lemma~3 in
  \cite{EisenbrandShmonin2006}] \label{lem:ipcaratheodory}
  Let~$X \subseteq \mathbb{Z}^d_{\geq 0}$ be a finite set of
  non-negative integer vectors and let~$b = (b_1,\ldots,b_d)$ be a
  vector in its integer conic hull~$\mathrm{intcone}(X)$.
  If~$|X| > \sum_{i=1}^d \log(b_i + 1)$, then there exists a proper
  subset~$X_0 \subseteq X$ such that~$b$ is in the integer conic
  hull~$\mathrm{intcone}(X_0)$ of~$X_0$.
\end{lemma}

The plan is to apply Lemma~\ref{lem:ipcaratheodory} on the set~$X$ of
column vectors of the constraint-matrix~$A$ of an integer linear
program along the lines of that in~\eqref{eqn:lpfortwo}, but
generalized to any number of bags. Precisely, with each
collection~$R_1(X_1),\ldots,R_m(X_m)$ of bags, we associate a linear program,
denoted by~$P(R_1,\ldots,R_m)$, that is a direct generalization
of the linear program in~\eqref{eqn:lpfortwo}. Let~$J = R'_1 \Join \cdots \Join R'_m$ be the
join of the supports of~$R_1,\ldots,R_m$. For each~$t \in J$, the linear program~$P(R_1,\ldots,R_m)$
has a variable~$x_t$. For each~$t \in J$, each~$i \in [m]$, and
each~$r \in R'_i$, define~$a_{r,t} = 1$ if~$t[X_i] = r$
and~$a_{r,t} = 0$ if~$t[X_i] \not= r$. Then, the constraints
of~$P(R_1,\ldots,R_m)$ are
\begin{equation}
\begin{array}{lll}
\sum_{t \in J} a_{r,t} x_t = R_i(r) & & \text{ for $i \in [m]$, $r \in R'_i$, }\\
x_t \geq 0 & & \text{ for $t \in J$.}
\end{array} \label{eqn:lpformany}
\end{equation}
Writing the equations of~$P(R_1,\ldots,R_m)$ in matrix form as~$Ax=b$,
it is important to note that, unless~$m = 2$, the matrix~$A$ is no
longer the vertex-edge incidence matrix of a bipartite graph as it was
when~$m = 2$. This means that the matrix~$A$ is no longer necessarily
totally unimodular. This point notwithstanding, the fact that the
integral solutions of~$P(R_1,\ldots,R_m)$ are still in 1-to-1
correspondence with the bags that witness the global consistency
of~$R_1,\ldots,R_m$ is all we need. We elaborate on this in the
next result. Before stating the result, we need the following additional concepts.

Let~$R_1,\ldots,R_m$ be globally consistent bags. If~$W$ is a bag that
witnesses the global consistency of~$R_1,\ldots,R_m$, then we say
that~$W$ is a \emph{minimal witness}
if there is no other bag~$U$ that witnesses the global consistency
of~$R_1,\ldots,R_m$ and is such that the support~$U'$ of~$U$ is
strictly contained in the support~$W'$ of~$W$. For a bag~$R$, define
\begin{itemize} \itemsep=0pt
\item its \emph{support size} by $\suppnorm{R} := |R'|$;
\item its \emph{multiplicity bound}
  by~$\munorm{R} := \max\{ R(r) : r \in R' \}$;
\item its \emph{multiplicity size}
  by~$\mbnorm{R} := \max\{ \log(R(r)+1) : r \in R' \}$;
\item its \emph{unary size} by~$\unorm{R} := \sum_{r \in R'} R(r)$;
\item its \emph{binary size}
by~$\bnorm{R} := \sum_{r \in R'} \log(R(r)+1)$.
\end{itemize}
Clearly, for every bag~$R$, the
inequalities~$\unorm{R} \leq \suppnorm{R} \munorm{R}$
and~$\bnorm{R} \leq \suppnorm{R} \mbnorm{R}$ hold.

\begin{theorem} \label{thm:bound}
  Let~$R_1,\ldots,R_m$ be globally consistent bags and let $W$ be a
  bag that witnesses their global consistency. Then the following
  statements are true.
\begin{enumerate}  \itemsep=0pt
\item $\munorm{W} \leq \max\{\munorm{R_i} : i \in [m] \}$.
\item $\suppnorm{W} \leq \sum_{i=1}^m \unorm{R_i}$.
\item If $W$ is a minimal witness, then $\suppnorm{W} \leq \sum_{i=1}^m \bnorm{R_i}$.
\end{enumerate}
\end{theorem}

\begin{proof}
  Let~$X_1,\ldots,X_m$ be the schemas of~$R_1,\ldots,R_m$. The first
  two statements follow from the fact that if~$W$ is a witness of the
  global consistency of~$R_1,\ldots,R_m$, then the
  equality
  \begin{equation}
  W(r) = \sum_{t \in W':\atop t[X_i]=r} W(t) = R_i(r)
  \end{equation}
   holds for
  each~$i \in [m]$ and each~$X_i$-tuple~$r$, and the
  quantities~$R_i(r)$ and~$W(t)$ with~$t \in W'$ are non-negative
  integers.

  For the third statement, assume that~$W$ is a minimal witness to the global
  consistency of~$R_1,\ldots,R_m$.
  Setting~$J := R'_1 \Join \cdots \Join R'_m$, by
  Lemma~\ref{lem:inclusion}, we have~$W' \subseteq J$. For
  each~$t \in J$, define~$x_t := W(t)$ and let~$x = (x_t : t \in
  J)$. It follows from the definitions that the vector~$x$ is an
  integer feasible solution for~$P(R_1,\ldots,R_m)$.  Write the
  equations of~$P(R_1,\ldots,R_m)$ in matrix form as~$Ax=b$, where~$A$
  is a~$d \times |J|$ matrix of zeros and ones
  where~$d := \sum_{i=1}^m |R'_i|$, and~$b \in \mathbb{Z}^d_{\geq 0}$
  is a~$d$-dimensional column vector with non-negative integer
  entries~$(R_i(r) : i \in [m], r \in R'_i)$.
  Let~$X = \{ c_t : t \in W' \}$ be the subset of the~$d$-dimensional
  column vectors of~$A$ that correspond to the non-zero components
  of~$x$. From the definition of~$P(R_1,\ldots,R_m)$ it follows that
  for every two distinct~$t,t' \in W'$, we have~$c_t \not= c_{t'}$.
  Hence~$|X|=|W'|$.

  The fact that~$Ax = b$ means that~$\sum_{t \in W'} c_t x_t = b$ and
  therefore the vector~$b$ belongs to the integer conic
  hull~$\intcone(X)$ of~$X$. Likewise, for every subset~$Q \subseteq W'$
  such that~$b$ is in the integer conic hull
  of~$X_0 := \{ c_t : t \in Q \}$, there exists a bag~$W_0$ with
  support~$Q$ that witnesses the global consistency
  of~$R_1,\ldots,R_m$. Therefore, since~$W$ is a minimal witness, it
  follows from Lemma~\ref{lem:ipcaratheodory}
  that~$|X| \leq \sum_{i=1}^m \sum_{r \in R'_i} \log(R_i(r)+1) =
  \sum_{i=1}^m \bnorm{R_i}$. Since~$|W'| = |X|$, the third statement has been
  proved.
\end{proof}

It should be noted that, assuming that all the numbers that are fed
into an algorithm are represented with the same number of bits by
adding leading zeros when necessary, the size of the representation of
a bag~$R$ when it is fed into an algorithm is~$\suppnorm{R}\munorm{R}$
when the multiplicities are represented in unary,
and~$\suppnorm{R}\mbnorm{R}$ when the multiplicities are represented
in binary. Therefore, since every globally consistent collection of
bags has a minimal witness of their global consistency,
Theorem~\ref{thm:bound} readily implies the following result.

\begin{corollary} \label{cor:inNP} The \gcpb~is in NP. \end{corollary}

It is worth noting that the first two statements of
Theorem~\ref{thm:bound} alone already imply the same if the
multiplicities of the given bags are bounded, or if they are
represented in unary. Nonetheless, as the following example shows, the
third statement of Theorem~\ref{thm:bound} is unavoidable if the
multiplicities are represented in binary, even if the schemas form
acyclic hypergraphs.

\begin{example}
  Consider
  bags~$R_1(A_1A_2), R_2(A_2A_3), \ldots, R_{n-1}(A_{n-1}A_n)$ with
  supports~$\{0,1\}^2$ and multiplicity~$2^n$ for each
  tuples in their support. Let~$J$ be the bag of schema~$A_1 \cdots A_n$,
  support~$\{0,1\}^n$, and multiplicity~$4$ for each tuple in its support. Then we
  have~$J[A_iA_{i+1}] = R_i$ for all~$i = 1,\ldots,n-1$, and~$|J'|$
  has cardinality~$2^n$, which is exponentially bigger than the
  size~$4(n-1)(n+1)$ of the input~$R_1,\ldots,R_{n-1}$, when the
  multiplicities are written in binary.
\end{example}

Theorem~\ref{thm:BFMY-bags} and Lemma \ref{lem:two-cons} imply that
the \gcpb~is solvable in polynomial time when restricted to acyclic
hypergraphs, since, in this case, global consistency of bags is
equivalent to pairwise consistency of bags, and the latter is
checkable in polynomial time.

We now turn to fixed hypergraphs, and state and prove the main
result of this section.

\begin{theorem} \label{thm:complexity}
Let $H=(V,E)$ be a hypergraph. Then the following statements are true.
\begin{enumerate} \itemsep=0pt
\item If $H$ is acyclic, then \glcpb$(H)$ is solvable in polynomial time.
\item If $H$ is cyclic, then \glcpb$(H)$ is NP-complete.
\end{enumerate}
\end{theorem}

\begin{proof}
  The first part of the theorem follows from
  Theorem~\ref{thm:BFMY-bags} and Lemma~\ref{lem:two-cons}. To prove
  the second part of the theorem, first note that membership in NP is
  a special case of Corollary~\ref{cor:inNP}.
  To prove NP-hardness, we will show that if~$H$ is a minimal
  non-chordal hypergraph or a minimal non-conformal hypergraph, then
  \glcpb$(H)$ is NP-complete.  More precisely, we will show in
  Lemmas~\ref{lem:cycles} and~\ref{lem:cliques} that both problems
  \glcpb$(C_n)$ and \glcpb$(H_n)$ are NP-complete for any~$n \geq
  3$. The desired NP-hardness will then follow from
  Lemmas~\ref{lem:characconf} and~\ref{lem:cons-preserv}.

\begin{lemma} \label{lem:cycles} For every $n\geq 3$, the
    problem~\glcpb$(C_n)$ is NP-complete.
\end{lemma}

\begin{proof}
  The problem \glcpb$(C_3)$ generalizes the problem of consistency of
  3-dimensional contingency tables (3DCT) from
  \cite{DBLP:journals/siamcomp/IrvingJ94}: given a positive
  integer~$n$ and, for each~$i,j,k \in [n]$, non-negative integer
  values~$R(i,k)$,~$C(j,k)$,~$F(i,j)$, is there
  an~$n \times n \times n$ table of non-negative integers~$X(i,j,k)$
  such that~$\sum_{q=1}^n X(i,q,k) =
  R(i,k)$,~$\sum_{q=1}^n X(q,j,k) =
  C(j,k)$,~$\sum_{q=1}^n X(i,j,q) = F(i,j)$ for
  all indices~$i,j,k \in [n]$?  To see this, let~$X,Y,Z$ be three
  attributes with domain~$[n]$, and let~$R(XZ)$,~$C(YZ)$,~$F(XY)$ be
  the three bags given by the three
  tables~$R(i,k)$,~$C(j,k)$,~$F(i,j)$. Therefore, \glcpb$(C_3)$ is
  NP-complete. For~$n \geq 4$, we show that there is a polynomial time
  reduction from \glcpb$(C_{n-1})$ to \glcpb$(C_n)$. The claim that
  \glcpb$(C_n)$ is NP-complete for every~$n \geq 3$ will follow by
  induction.

  Let~$R_1(A_1A_2),R_2(A_2A_3),\ldots,R_{n-1}(A_{n-1}A_1)$ be an
  instance of \glcpb$(C_{n-1})$. Let~$A_n$ be a new attribute with the
  same domain as~$A_1$. The reduction replaces the
  bag~$R_{n-1}(A_{n-1}A_1)$ by an identical copy~$R_{n-1}(A_{n-1}A_n)$
  of schema~$A_{n-1}A_n$, and adds one more bag~$R_n(A_nA_1)$ with
  support~$R'_n = \{(a,a) : a \in \mathrm{Dom}(A_1) \}$, and
  multiplicities defined by~$R_n(a,a) = R_{n-1}(a)$ for
  every~$(a,a) \in R'_n$, where~$R_{n-1}(a)$ denotes the multiplicity
  of~$a$ in the~$R_{n-1}[A_1]$. If~$R$ is a bag that witnesses the
  global consistency of~$R_1,\ldots,R_{n-1}$, then the
  bag~$S(A_1 \cdots A_n)$ defined, for each~$A_1 \cdots A_n$-tuple~$t$
  by~$S(t) = R(t[A_1\cdots A_{n-1}])$ whenever~$t[A_n] = t[A_{n-1}]$
  and~$S(t) = 0$ otherwise, witnesses the global consistency
  of~$R_1,\ldots,R_n$. Conversely, if~$S$ is a bag that witnesses the
  global consistency of~$R_1,\ldots,R_n$, then the
  bag~$R(A_1 \cdots A_{n-1})$ defined, for
  each~$A_1 \cdots A_{n-1}$-tuple~$t$ by~$R(t) = R(t,t[A_{n-1}])$,
  witnesses the global consistency of~$R_1,\ldots,R_{n-1}$.
\end{proof}

\begin{lemma} \label{lem:cliques}
For every $n\geq 3$, the problem~\glcpb$(H_n)$ is NP-complete.
\end{lemma}

\begin{proof}
  Since~$H_3 = C_3$, the problem \glcpb$(H_3)$ is NP-complete by the
  previous lemma. For~$n \geq 4$, we show that there is a
  polynomial time reduction from \glcpb$(H_{n-1})$ to
  \glcpb$(H_n)$. The claim that \glcpb$(H_n)$ is NP-complete for
  every~$n \geq 3$ will follow by induction.

  Let~$R_1(X_1),\ldots,R_{n-1}(X_{n-1})$ be bags,
  where~$X_i = \{A_1,\ldots,A_{n-1}\}\setminus\{A_i\}$
  for~$i \in [n-1]$. Let~$A_n$ be a new attribute with
  domain~$\{1,2\}$ and define new bags~$S_1(Y_1),\ldots,S_n(Y_n)$
  with~$Y_i = \{A_1,\ldots,A_n\}\setminus\{A_i\}$ for~$i \in [n]$ as
  follows.  For~$i \in [n-1]$, let~$D_i$ be the size of the active
  domain of the attribute~$A_i$ in the supports~$R'_1,\ldots,R'_{n-1}$
  of~$R_1,\ldots,R_{n-1}$, and let~$M$ be the maximum of all
  multiplicities in~$R_1,\ldots,R_{n-1}$.  For~$i \in [n-1]$,
  define~$S_i(t,1) = R_i(t)$ and~$S_i(t,2) = MD_i - R_i(t)$ for
  any~$X_i$-tuple~$t$.  For~$i = n$, define~$S_i(t) = M$ for
  any~$Y_i$-tuple~$t$.  We claim that this reduction works. Indeed,
  given a witness~$R$ for the global consistency
  of~$R_1,\ldots,R_{n-1}$, we can produce a witness~$S$ for the global
  consistency of~$S_1,\ldots,S_n$ by setting~$S(t,1) = R(t)$
  and~$S(t,2) = M-R(t)$ for any~$A_1\cdots
  A_{n-1}$-tuple~$t$. Conversely, given a witness~$S$ for the global
  consistency of~$S_1,\ldots,S_n$, we can produce a witness~$R$ for
  the global consistency of~$R_1,\ldots,R_{n-1}$ by
  setting~$R(t) = S(t,1)$ for any~$A_1,\ldots,A_{n-1}$-tuple~$t$.
\end{proof}
The proof of Theorem \ref{thm:complexity} is now complete.
\end{proof}

\subsection{Finding the Witness}

In this section, we address the question of producing a small witness
to global consistency, when  a witness to global consistency exists. Theorem~\ref{thm:bound} ensures
that if there is any witness at all, then a small one exists, but it does not tell us how to construct a small witness.

We start by noting that, for any fixed cyclic hypergraph~$H$, one
cannot hope to find small witnesses to the global consistency of
given bags over~$H$ in time polynomial in the size of the input,
unless~P~=~NP. Indeed, just deciding if a witness exists is
already~NP-hard by Theorem~\ref{thm:complexity}. Since checking if a
witness is valid is a problem that can be solved in polynomial time,
the problem of finding a witness can only be harder.  For acyclic
hypergraphs, however, we will see that the structural results of
Section~\ref{sec:three-bag-cons} provide a way to construct a
witness. For this, we will need a strengthening of
Corollary~\ref{cor:two-cons-poly} to the effect that not only a
witness to the consistency of two bags can be found, but even a minimal
witness can be found in strongly polynomial time. We will also need a
strengthening of Theorem~\ref{thm:bound} in the special case of two
bags.

To describe the algorithm that finds minimal witnesses, we
need to introduce some terminology. Let~$R(X)$ and~$S(Y)$ be two bags
and consider the network~$N(R,S)$. In what follows, an edge~$(u,v)$
of~$N(R,S)$ of the form~$(t[X],t[Y])$ with~$t \in R' \Join S'$ is
called a \emph{middle edge}. The proof of Lemma~\ref{lem:two-cons}
established that if~$R$ and~$S$ are consistent and~$f(u,v)$ is a
saturated flow of the network~$N(R,S)$, then the bag~$T(XY)$ defined
by setting~$T(t) := f(t[X],t[Y])$ for each middle edge~$(t[X],t[Y])$
is a witness to the consistency of~$R$ and~$S$. In particular, the
support~$T'$ of the witness~$T$ is the set of middle edges of~$N(R,S)$
that are used by the flow~$f$.

In order to find a minimal witness to the consistency of~$R$ and~$S$,
we proceed by self-reducibility, deleting middle edges from~$N(R,S)$
one by one. We loop through the middle edges~$(u,v)$ of the current
network and, for each one, ask: is the middle edge~$(u,v)$ used by all
saturated flows of the current network?  If the answer is no, then it
is safe to delete the edge and continue with the new network. If the
answer is yes, then we keep the edge and proceed to the next middle
edge. To tell whether a middle edge~$(u,v)$ is used by all saturated
flows of the current network, we can temporarily remove it, compute a
maximum flow of the resulting network, and check whether it is
saturated. Since the number of middle edges of the initial
network~$N(R,S)$ is~$|R' \Join S'|$, a saturated flow along a minimal
subset of middle edges will be found after at most~$|R' \Join S'|$
many such tests. This gives a minimal witness for the consistency
of~$R$ and~$S$.

Before we state the strengthening of
Corollary~\ref{cor:two-cons-poly}, we also need to strengthen the
bound on the support-size of minimal witnesses given by
Theorem~\ref{thm:bound} in the case~$m = 2$.  For this special case,
the standard form of Carath\'eodory's Theorem will suffice. The
\emph{conic hull} of a set~$X\subseteq \mathbb{R}^d$, where~$d\geq 1$,
is the set of all vectors in~$\mathbb{R}^d$ that can be written as a
linear combination of vectors from~$X$ with non-negative
coefficients. Carath\'eodory's Theorem asserts that
if~$X\subseteq \mathbb{R}^d$ for some~$d\geq 1$ and if a vector~$x$
belongs to the conic hull of~$X$, then there is a subset~$X_0$ of~$X$
of cardinality at most~$d$ such that~$x$ belongs to the conic hull
of~$X_0$ (in Schrijver's book on linear and integer programming, this
is stated as Corollary~7.1i and it follows from more general results
about linear programming).

\begin{theorem}\label{thm:bound2}
  Let~$R$ and~$S$ be consistent bags and let~$W$ be a bag that
  witnesses their consistency. If~$W$ is a minimal witness to the
  consistency of~$R$ and~$S$,
  then~$\suppnorm{W} \leq \suppnorm{R} + \suppnorm{S}$.
\end{theorem}

\begin{proof}
  Let~$J := R' \Join S'$, so~$W' \subseteq J$ by
  Lemma~\ref{lem:inclusion}. By setting~$x_t := W(t)$ for
  each~$t \in J$,  we get a feasible solution for the linear
  program~$P(R,S)$. If we write the constraint matrix of~$P(R,S)$ in
  matrix form as~$Ax = b$, this means that the vector~$b$ is in the
  conic hull of the set of columns of~$A$ indexed by tuples~$t$
  in~$W'$. By Carath\'eodory's Theorem,~$b$ is also in the conic hull
  of a subset of at most~$d$ many of the columns of~$A$ indexed by
  tuples~$t$ in~$W'$, where~$d := \suppnorm{R} + \suppnorm{S}$ is the
  dimension of the vector~$b$. This means that there
  exists~$J_0 \subseteq W'$ with~$|J_0| \leq d$ and a non-negative
  vector~$y = (y_t : t \in J)$ with~$y_t = 0$ for
  each~$t \in J \setminus J_0$ such that~$Ay = b$.
  Setting~$f(t[X],t[Y]) = y_t$ for each~$t \in J_0$, we get a saturated
  flow of the subnetwork~$N_0$ of~$N(R,S)$ in which all middle edges
  of the form~$(t[X],t[Y])$ with~$t \in J \setminus J_0$ have been
  supressed.  Since all capacities of~$N(R,S)$ are integers, as in the
  proof of Lemma~\ref{lem:two-cons}, the integrality theorem for the
  max-flow problem gives a max flow~$f_0(u,v)$ of~$N_0$ with
  integers. This flow of~$N_0$ is also saturated, which means that by
  setting~$W_0(t) := f_0(t[X],t[Y])$ for each~$t \in J_0$ we get a witness
  of the consistency of~$R$ and~$S$ with support~$W_0'$ included
  in~$J_0 \subseteq W'$. Since~$W$ is minimal we have~$J_0 = W'$, from
  which it follows that~$|W'| = |J_0| \leq d$. That
  is,~$\suppnorm{W} \leq \suppnorm{R}+\suppnorm{S}$.
\end{proof}

\begin{corollary} \label{cor:stronger} There is a strongly
  polynomial-time algorithm that, given two bags~$R$ and~$S$,
  determines whether they are consistent and, if they are, constructs
  a bag~$T$ that is a minimal witness of their consistency. In
  particular,~$\suppnorm{T} \leq \suppnorm{R} + \suppnorm{S}$.
\end{corollary}

We now put everything together to show that, over acyclic schemas, a
witness to global consistency can be found in polynomial time.

\begin{theorem}
  There is a polynomial time algorithm that, given an acyclic
  hypergraph~$H$ and a collection of bags over~$H$, determines whether
  the collection is globally consistent and, if it is, constructs a
  bag that is a witness to the global consistency of the
  collection. Furthermore, the bag that the algorithm returns has its
  support-size bounded by the sum of the support-sizes of the input
  bags.
\end{theorem}

\begin{proof}
  Let~$H = (V,\{X_1,\ldots,X_m\})$ be an acyclic hypergraph and
  let~$R_1(X_1),\ldots,R_m(X_m)$ be a collection of bags over~$H$.
  First, we test for pairwise consistency. If there are two bags in
  the collection that are not consistent, then the collection cannot
  be globally consistent, and we stop. Otherwise, we proceed as in the
  proof of Theorem~\ref{thm:BFMY-bags} to construct a witness of their
  global consistency as follows.

  By first computing a rooted join-tree in polynomial time (see
  \cite{10.1137/0213035}) and then by sorting its vertices in
  topological order, we may assume that the listing~$X_1,\ldots,X_m$
  satisfies the running intersection property: for every~$i \in [m]$
  with~$i \geq 2$, there is a~$j \in [i-1]$ such
  that~$X_i \cap (X_1 \cup \cdots \cup X_{i-1}) \subseteq X_j$. By
  induction on~$i = 1,\ldots,m$, we construct a bag~$T_i$
  over~$X_1 \cup \cdots \cup X_i$ that is a witness to the global
  consistency of the bags~$R_1,\ldots,R_i$ and
  satisfies~$\suppnorm{T_i} \leq \sum_{j=1}^i \suppnorm{R_i}$.
  For~$i = 1$, we take~$T_i = R_i$.  For~$i \geq 2$, we apply the
  algorithm given by Corollary~\ref{cor:stronger} on the
  bags~$T_{i-1}$ and~$R_i$ to obtain~$T_i$. In Step~1 of the proof of
  Theorem~\ref{thm:BFMY-bags}, we showed that any bag that witnesses
  the consistency of~$T_{i-1}$ and~$R_i$, such as~$T_i$, also
  witnesses the global consistency of~$R_1,\ldots,R_i$. By
  Corollary~\ref{cor:stronger}, we also
  have that~$\suppnorm{T_i} \leq \suppnorm{T_{i-1}} + \suppnorm{R_i}$, from
  which the desired bound~$\suppnorm{T_i} \leq \sum_{j=1}^i \suppnorm{R_j}$
  follows by the induction hypothesis.

  Let~$M$ be the maximum multiplicity in the input
  bags~$R_1,\ldots,R_m$ and let~$B = \log(M+1)$ be the number of bits
  it takes to represent them. By Theorem~\ref{thm:bound} we have that
  all the multiplicities of every~$T_i$ are bounded by~$M$.
  Therefore, the size of each $T_i$ is bounded
  by~$B\suppnorm{T_i} \leq B\sum_{j=1}^i \suppnorm{R_j}$. The runtime
  of the algorithm is then bounded by~$m$ times the runtime of the
  algorithm in Corollary~\ref{cor:stronger} on inputs of these sizes,
  and is thus bounded by a polynomial
  in~$B\sum_{j=1}^m \suppnorm{R_j}$, i.e., the size of the input.
\end{proof}

\section{Concluding Remarks} \label{sec:conclusions}

In this paper, we investigated the interplay between local consistency
and global consistency for bags. At the structural level, we showed
that bags behave like relations as regards the \ltgc, namely, the
\ltgc~for bags holds over a schema if and only if the sets of
attributes of that schema form an acyclic hypergraph. At the
algorithmic level, however, bags behave different than relations as
regards testing for global consistency. Specifically, for every fixed
schema, testing relations for global consistency is solvable in
polynomial time, while for bags this happens precisely when the schema
is acyclic - otherwise, testing bags for global consistency is
NP-complete.

We conclude by describing certain open problems that are motivated by
the work reported here.

Beeri et al.\ \cite{BeeriFaginMaierYannakakis1983} showed that
hypergraph acyclicity is also equivalent to certain semantic
conditions other than the \ltgc~for relations, including the existence
of a \emph{full reducer} and the existence of a \emph{monotone
  sequential join expression}. Do analogous results hold bags? One of
the difficulties in answering this question is that the bag-join of
two consistent relations need not witness their consistency, thus it
is not at all clear how to define a suitable semi-join operation for
bags or how to find a suitable substitute for a monotone sequential
join expression.

As mentioned in the Introduction, we have recently studied a relaxed
notion of consistency for~$K$-relations, where~$K$ is a positive
semiring \cite{DBLP:journals/corr/abs-2009-09488}. The goal of that
investigation was to find a common generalization of the results by
Vorob'ev \cite{vorob1962consistent} and by Beeri et al.\
\cite{BeeriFaginMaierYannakakis1983}. The stricter notion of
consistency for bags studied here makes perfectly good sense
for~$K$-relations as well. It is an open problem whether or not the
results presented here extend to~$K$-relations under the stricter
notion of consistency, where~$K$ is a positive semiring or some other
type of semiring for which there is a good theory for solving systems
of linear equations or other combinatorial problems formulated over
that semiring.

\bibliographystyle{alpha}
\bibliography{biblio}

\end{document}